\pdfoutput=1

\RequirePackage{amsmath}
\documentclass[runningheads,envcountsame,envcountsect]{llncs}
\usepackage[crop]{llncsconf}

\usepackage{hyperref,xcolor}

\renewcommand{\orcidID}[1]{} 

\let\doendproof\endproof
\renewcommand\endproof{~\hfill$\qed$\doendproof} 

\usepackage{color}
\definecolor{darkred}{rgb}{0.7,0,0}

\newcommand{\fig}{./fig}
\newcommand{\figref}[1]{\figurename~\ref{#1}}

\newtheorem{observation}[theorem]{Observation}

\usepackage{thmtools}
\usepackage{thm-restate}

\usepackage{multirow}

\usepackage{tabularx,environ}
\makeatletter
\newcommand{\problemtitle}[1]{\gdef\@problemtitle{#1}}
\newcommand{\probleminput}[1]{\gdef\@probleminput{#1}}
\newcommand{\problemquestiontitle}[1]{\gdef\@problemquestiontitle{#1}}
\newcommand{\problemquestion}[1]{\gdef\@problemquestion{#1}}
\NewEnviron{myproblem}{
  \problemtitle{}\probleminput{}\problemquestion{}
  \BODY
  \par\addvspace{.5\baselineskip}
  \noindent \normalsize
  \begin{tabularx}{\textwidth}{@{\hspace{\parindent}} l X c}
    \multicolumn{2}{@{\hspace{\parindent}}l}{\normalsize\@problemtitle} \\
    \normalsize \textbf{Input:} & \normalsize \  \@probleminput \\
    \normalsize \textbf{\@problemquestiontitle:} & \normalsize \  \@problemquestion
  \end{tabularx}
  \par\addvspace{.5\baselineskip}
}
\makeatother

\newcommand{\cw}{\mathsf{cw}} 
\newcommand{\tw}{\mathsf{tw}} 
\newcommand{\pw}{\mathsf{pw}} 
\newcommand{\td}{\mathsf{td}} 
\newcommand{\vi}{\mathsf{vi}} 
\newcommand{\vc}{\mathsf{vc}} 

\newcommand{\bw}{\mathsf{bw}} 

\newcommand{\im}{\mathsf{im}} 

\newcommand{\dist}{\operatorname{dist}}

\newcommand{\indsub}{\preceq_{\mathrm{I}}} 

\usepackage{mathtools}

\usepackage{amssymb}

\titlerunning{Treedepth, Vertex Cover, and Vertex Integrity}

\title{Exploring the Gap Between Treedepth and Vertex Cover Through Vertex Integrity\thanks{%
Partially supported
by JSPS KAKENHI Grant Numbers 
JP18H04091, 
JP18K11168, 
JP18K11169, 
JP19K21537, 
JP20K19742, 
JP20H05793. 
}}

\author{
Tatsuya Gima\inst{1} \and
Tesshu Hanaka\inst{2}\orcidID{0000-0001-6943-856X} \and
Masashi Kiyomi\inst{3}\orcidID{0000-0003-1618-9373} \and
Yasuaki Kobayashi\inst{4}\orcidID{0000-0003-3244-6915} \and
Yota Otachi\inst{1}\orcidID{0000-0002-0087-853X}
}
\institute{
Nagoya University, Nagoya, Japan\\
\email{gima@nagoya-u.jp}, \email{otachi@nagoya-u.jp}
\and
Chuo University, Bunkyo-ku, Tokyo, Japan\\
\email{hanaka.91t@g.chuo-u.ac.jp}
\and
Yokohama City University, Yokohama, Japan\\
\email{masashi@yokohama-cu.ac.jp}
\and
Kyoto University, Kyoto, Japan\\
\email{kobayashi@iip.ist.i.kyoto-u.ac.jp}
}
\authorrunning{Gima et al.}

\usepackage{lineno}
\newcommand*\patchAmsMathEnvironmentForLineno[1]{
  \expandafter\let\csname old#1\expandafter\endcsname\csname #1\endcsname
  \expandafter\let\csname oldend#1\expandafter\endcsname\csname end#1\endcsname
  \renewenvironment{#1}
     {\linenomath\csname old#1\endcsname}
     {\csname oldend#1\endcsname\endlinenomath}}
\newcommand*\patchBothAmsMathEnvironmentsForLineno[1]{
  \patchAmsMathEnvironmentForLineno{#1}
  \patchAmsMathEnvironmentForLineno{#1*}}
\AtBeginDocument{
\patchBothAmsMathEnvironmentsForLineno{equation}
\patchBothAmsMathEnvironmentsForLineno{align}
\patchBothAmsMathEnvironmentsForLineno{flalign}
\patchBothAmsMathEnvironmentsForLineno{alignat}
\patchBothAmsMathEnvironmentsForLineno{gather}
\patchBothAmsMathEnvironmentsForLineno{multline}
}

\begin{document}

\maketitle

\begin{abstract}
For intractable problems on graphs of bounded treewidth,
two graph parameters treedepth and vertex cover number have been used to obtain fine-grained complexity results.
Although the studies in this direction are successful,
we still need a systematic way for further investigations because the graphs of bounded vertex cover number
form a rather small subclass of the graphs of bounded treedepth.
To fill this gap, we use vertex integrity, which is placed between the two parameters mentioned above.
For several graph problems, we generalize fixed-parameter tractability results
parameterized by vertex cover number to the ones parameterized by vertex integrity.
We also show some finer complexity contrasts by showing hardness with respect to vertex integrity or treedepth.
\end{abstract}


\section{Introduction}

Treewidth, which measures how close a graph is to a tree,
is arguably one of the most powerful tools for designing efficient algorithms for graph problems.
The application of treewidth is quite wide and the general theory built there often gives
a very efficient algorithm (e.g.,~\cite{Bodlaender88,ArnborgLS91,Courcelle92}).
However, still many problems are found to be intractable on graphs of bounded treewidth (e.g.,~\cite{Szeider11arxiv}).
To cope with such problems, one may use pathwidth, which is always larger than or equal to treewidth.
Unfortunately, this approach did not quite work
as no natural problem was known to change its complexity with respect to treewidth and pathwidth,
until very recently~\cite{BelmonteKLMO20}.
Treedepth is a further restriction of pathwidth.
However, still most of the problems do not change their complexity,
except for some problems with hardness depending on the existence of long paths
(e.g.,~\cite{DvorakK18,KellerhalsK20}).
One successful approach in this direction is parameterization by the vertex cover number,
which is a strong restriction of treedepth.
Many problems that are intractable parameterized by treewidth
have been shown to become tractable when parameterized by
vertex cover number~\cite{FellowsLMRS08,EncisoFGKRS09,FialaGK11,Abu-Khzam14,Lokshtanov15,BonnetS17}.

One drawback of the vertex-cover parameterization is its limitation to a very small class of graphs.
To overcome the drawback, we propose a new approach for parameterizing graph problems by vertex integrity~\cite{BarefootES87}.
The \emph{vertex integrity} of a graph $G$, denoted $\vi(G)$, 
is the minimum integer $k$ satisfying that
there is $S \subseteq V(G)$ such that $|S| + |V(C)| \le k$ for each component $C$ of $G-S$.
We call such $S$ a \emph{$\vi(k)$-set} of $G$.
This parameter is bounded from above by vertex cover number${}+1$ and from below by treedepth.
As a structural parameter in parameterized algorithms,
vertex integrity (and its close variants) was used only in a couple of previous studies~\cite{DvorakEGKO17,GanianKO18,BodlaenderHOOZ19}.
Our goal is to fill some gaps between treedepth and vertex cover number
by presenting finer algorithmic and complexity results parameterized by vertex integrity.
Note that the parameterization by vertex integrity is equivalent to the one by $\ell$-component order connectivity${}+\ell$~\cite{DrangeDH16}.

\medskip

\noindent\textit{Short preliminaries.}
For the basic terms and concepts in the parameterized complexity theory,
we refer the readers to standard textbooks, e.g.~\cite{DowneyF99,CyganFKLMPPS15}.

For a graph $G$, we denote 
its treewidth by $\tw(G)$,
pathwidth by $\pw(G)$,
treedepth by $\td(G)$, and
vertex cover number by $\vc(G)$.
(See Section~\ref{sec:graph-parameters} for definitions.)
It is known that $\tw(G) \le \pw(G) \le \td(G)-1 \le \vi(G)-1 \le \vc(G)$ for every graph $G$.
We say informally that a problem is fixed-parameter tractable ``parameterized by $\vi$'',
which means ``parameterized by the vertex integrity of the input graphs.''
We also say ``graphs of $\vi = c$ (or $\vi \le c$)''.

\medskip

\noindent\textit{Our results.}
The main contribution of this paper is to generalize several known FPT algorithms parameterized by $\vc$ to the ones by $\vi$.
We also show some results considering parameterizations by $\vc$, $\vi$, or $\td$
to tighten the complexity gaps between parameterizations by $\vc$ and by $\td$.
See Table~\ref{tbl:summary} for the summary of results.
Due to the space limitation, we had to move most of the results into the appendix.
In the main text, we present full descriptions of selected results only.
(Even for the selected results, we still have to omit some proofs. They are marked with $\bigstar$.)

\textit{Extending FPT results parameterized by $\vc$.}
We show that
\textsc{Imbalance},
\textsc{Maximum Common (Induced) Subgraph},
\textsc{Capacitated Vertex Cover},
\textsc{Capacitated Dominating Set},
\textsc{Precoloring Extension},
\textsc{Equitable Coloring}, and
\textsc{Equitable Connected Partition}
are fixed-parameter tractable parameterized by vertex integrity.
We present the algorithms for \textsc{Imbalance} as a simple but still powerful example that generalizes known results (Section~\ref{sec:imbalance})
and for \textsc{Maximum Common Subgraph} as one of the most involved examples (Section~\ref{sec:mcs}).
See Section~\ref{sec:extending-vc} for the other problems.
A commonly used trick is to reduce the problem instance to a number of instances of integer linear programming,
while each problem requires a nontrivially tailored reduction depending on its structure.
It was the same for parameterizations by $\vc$,
but the reductions here are more involved because of the generality of $\vi$.
Finding the similarity among the reductions and algorithms would be a good starting point to develop a general way for handling problems
parameterized by $\vi$ (or $\vc$).
Additionally, we show that \textsc{Bandwidth} is W[1]-hard parameterized by $\td$,
while we were not able to extend the algorithm parameterized by $\vc$ to the one by $\vi$.

\textit{Filling some complexity gaps.}
We observe that \textsc{Graph Motif} and \textsc{Steiner Forest}
have different complexity with respect to $\vc$ and $\vi$ (Section~\ref{sec:hard-vi}).
In particular, we see that not all FPT algorithms parameterized by $\vc$ can be generalized to the ones by $\vi$.
\textsc{Min Max Outdegree Orientation} gives an example that a known hardness for $\td$ can be strengthened to the one for $\vc$
(Section~\ref{sec:min-max-outdeg}).
We additionally observe that some W[1]-hard problems parameterized by $\tw$ 
become tractable parameterized by $\td$. Such problems include \textsc{Metric Dimension}, \textsc{Directed} $(p,q)$-\textsc{Edge Dominating Set}, and \textsc{List Hamiltonian Path} (Section~\ref{sec:easy-td}).

\begin{table}[bt]
  \centering
  \caption{Summary. The results stated without references are shown in this paper.}
  \begin{tabular}{l|l|l}
    \textsc{Problem} & Lower bounds & Upper bounds\\\hline
    \multirow{2}{*}{\textsc{Imbalance}} & \multirow{2}{*}{NP-h \cite{BiedlCGHW05}} & FPT by $\tw + \Delta$ \cite{LokshtanovMS13}\\
    & & FPT by $\vi$\\\hline
    \textsc{Max Common Subgraph} & NP-h for $\vi(G_2) = 3$ & \multirow{2}{*}{FPT by $\vi(G_1) + \vi(G_2)$}\\
    \textsc{Max Common Ind.\ Subgraph} & NP-h for $\vc(G_2) = 0$ &\\\hline
    \textsc{Capacitated Vertex Cover} & W[1]-h by $\td$ \cite{DomLSV08} & FPT by $\vi$\\\hline
    \textsc{Capacitated Dominating Set} & W[1]-h by $\td + k$ \cite{DomLSV08} & FPT by $\vi$\\\hline
    \textsc{Precoloring Extension} & W[1]-h by $\td$ \cite{FellowsFLRSST11} & FPT by $\vi$\\\hline
    \textsc{Equitable Coloring} & W[1]-h by $\td$ \cite{FellowsFLRSST11} & FPT by $\vi$\\\hline
    \textsc{Equitable Connected Part.} & W[1]-h by $\pw$ \cite{EncisoFGKRS09} & FPT by $\vi$\\\hline
    \multirow{2}{*}{\textsc{Bandwidth}} & W[1]-h by $\td$ & FPT by $\vc$ \cite{FellowsLMRS08} \\ 
    &NP-h for $\pw = 2$ \cite{Muradian03} & P for $\pw \le 1$~\cite{AssmannPSZ81}\\\hline
    \multirow{2}{*}{\textsc{Graph Motif}} & \multirow{2}{*}{NP-h for $\vi = 4$} & FPT by $\vc$ \cite{BonnetS17}\\
    && P for $\vi \le 3$\\\hline
    \textsc{Steiner Forest} & NP-h for $\vi = 5$ \cite{Gassner10} & XP by $\vc$\\
    \textsc{Unweighted Steiner Forest} & NP-h for $\tw = 3$ \cite{Gassner10} & FPT by $\vc$\\\hline
    \textsc{Unary Min Max Outdeg.\ Ori.} & W[1]-h by $\vc$ & XP by $\tw$ \cite{Szeider11} \\
    \textsc{Binary Min Max Outdeg.\ Ori.} & NP-h for $\vc = 3$ & P for $\vc \le 2$ \\\hline
    \multirow{2}{*}{\textsc{Metric Dimension}} & \multirow{2}{*}{W[1]-h by $\pw$ \cite{BonnetP19}} & FPT by $\tw + \Delta$ \cite{BelmonteFGR17}\\
    && FPT by $\td$\\\hline
    \multirow{2}{*}{\textsc{Directed $(p, q)$-Edge Dom.\ Set}} & \multirow{2}{*}{W[1]-h by $\pw$ \cite{BelmonteHK0L18}} & FPT by $\tw + p + q$ \cite{BelmonteHK0L18}\\
    && FPT by $\td$\\\hline
    \textsc{List Hamiltonian Path} & W[1]-h by $\pw$ \cite{MeeksS16} & FPT by $\td$
  \end{tabular}
  \label{tbl:summary}
\end{table}


\section{\textsc{Imbalance}}
\label{sec:imbalance}

In this section, we show that \textsc{Imbalance} is fixed-parameter tractable parameterized by $\vi$.
Let $G = (V,E)$ be a graph.
Given a linear ordering $\sigma$ on $V$,
the \emph{imbalance} $\im_{\sigma}(v)$ of $v \in V$
is the absolute difference of the numbers of the neighbors of $v$
that appear before $v$ and after $v$ in $\sigma$.
The \emph{imbalance} of $G$, denoted $\im(G)$,
is defined as $\min_{\sigma} \sum_{v \in V} \im(v)$,
where the minimum is taken over all linear orderings on $V$.
Given a graph $G$ and an integer $b$, 
\textsc{Imbalance} asks whether $\im(G) \le b$.

Fellows et al.~\cite{FellowsLMRS08} showed that \textsc{Imbalance} is fixed-parameter tractable parameterized by $\vc$.
Recently, Misra and Mittal~\cite{MisraM20} have extended the result by showing that \textsc{Imbalance} is fixed-parameter tractable parameterized by
the sum of the twin-cover number and the maximum twin-class size.
Although twin-cover number is incomparable with vertex integrity,
the combined parameter in~\cite{MisraM20} is always larger than or equal to the vertex integrity of the same graph.
On the other hand, the combined parameter can be arbitrarily large for some graphs of constant vertex integrity
(e.g., disjoint unions of $P_{3}$'s).
Hence, our result here properly extends the result in~\cite{MisraM20} as well.

\smallskip

\textit{Key concepts.}
Before proceeding to the algorithm,
we need to introduce two important concepts 
that are common in our algorithms parameterized by $\vi$.

1. \textit{ILP parameterized by the number of variables.}
It is known that the feasibility of an instance of integer linear programming (ILP) 
parameterized by the number of variables is fixed-parameter tractable~\cite{Lenstra83}.
Using the algorithm for the feasibility problem as a black box, 
one can show the same fact for the optimization version as well.
(See Section~\ref{sec:ilp} for the detail.)
This fact has been used heavily 
for designing FPT algorithms parameterized by $\vc$ (see e.g.~\cite{FellowsLMRS08}).
We are going to see that some of these algorithms can be generalized for the parameterization by $\vi$,
and \textsc{Imbalance} is the first such example.

2. \textit{Equivalence relation among components.}
For a vertex set $S$ of $G$,
we define an equivalence relation $\sim_{G,S}$ among components of $G-S$
by setting $C_{1} \sim_{G,S} C_{2}$
if and only if
there is an isomorphism $g$ from $G[S \cup V(C_{1})]$ to $G[S \cup V(C_{2})]$
that fixes $S$; that is, $g|_{S}$ is the identity function.
When $C_{1} \sim_{G,S} C_{2}$, we say that $C_{1}$ and $C_{2}$ have the same \emph{$(G, S)$-type}
(or just the same \emph{type} if $G$ and $S$ are clear from the context). 
See \figref{fig:type}.
We say that a component $C$ of $G - S$ is of \emph{$(G, S)$-type $t$} (or just \emph{type $t$})
by using a canonical form $t$ of the members of the $(G, S)$-type equivalence class of $C$.
We can set the canonical form $t$ in such a way that it can be computed from $S$ and $C$ in time depending only on
$|S \cup V(C)|$.\footnote{For example, by fixing the ordering of vertices in $S$ as $v_{1}, \dots, v_{|S|}$,
we can set $t$ to be the adjacency matrix of $G[S \cup V(C)]$
such that the $i$th row and column correspond to $v_{i}$ for $1 \le i \le |S|$
and under this condition the string $t[1,1], \dots, t[1, s], t[2,1], \dots, t[s,s]$
is lexicographically minimal, where $s = |S \cup V(C)|$.}
Observe that if $S$ is a $\vi(k)$-set of $G$,
then the number of $\sim_{G,S}$ classes depends only on $k$
since $|S \cup V(C)| \le k$ for each component $C$ of $G-S$.
Hence, we can compute for all types $t$
the number of type-$t$ components of $G-S$ in $O(f(k) \cdot n)$ total running time,
where $n = |V|$ and $f(k)$ is a computable function depending only on $k$.
Note that this information (the numbers of type-$t$ components for all $t$) 
completely characterizes the graph $G$ up to isomorphism.

\begin{figure}[htb]
  \centering
  \includegraphics[scale=.8]{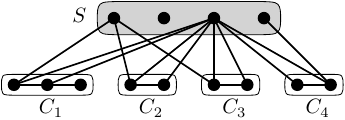} 
  \caption{The components $C_{2}$ and $C_{3}$ of $G-S$ have the same $(G,S)$-type.}
  \label{fig:type}
\end{figure}

\begin{theorem}
\label{thm:imb}
\textsc{Imbalance} is fixed-parameter tractable parameterized by $\vi$.
\end{theorem}
\begin{proof}
Let $S$ be a $\vi(k)$-set of $G$. Such a set can be found in $O(k^{k+1} n)$ time~\cite{DrangeDH16}.
We first guess and fix the relative ordering of $S$ in an optimal ordering. 
There are only $k!$ candidates for this guess.
For each $v \in S$, let $\ell(v)$ and $r(v)$ be the numbers of vertices in $N(v) \cap S$
that appear before $v$ and after $v$, respectively, in the guessed relative ordering of $S$.

Observe that the imbalance of a vertex $v$ in a component $C$ of $G-S$ depends only on the relative ordering of $S \cup V(C)$
since $N(v) \subseteq S \cup V(C)$.
For each type $t$ and for each relative ordering $p$ of $S \cup V(C)$, where $C$ is a type-$t$ component of $G-S$,
we denote by $\im(t,p)$ the sum of imbalance of the vertices in $C$.
Similarly, the numbers of vertices in a type-$t$ component $C$
that appear before $v \in S$ and after $v$ depend only on the relative ordering $p$ of $S \cup V(C)$;
we denote these numbers by $\ell(v,t,p)$ and $r(v,t,p)$, respectively.
The numbers $\im(t,p)$, $\ell(v,t,p)$, and $r(v,t,p)$ can be computed from their arguments in time 
depending only on $k$, and thus they are treated as constants in the following ILP\@.

We represent by a nonnegative variable $x_{t,p}$ the number of type-$t$ components 
that have relative ordering $p$ with $S$. Note that the number of combinations of $t$ and $p$ depends only on $k$.
For each $v \in S$, we represent (an upper bound of) the imbalance of $v$ by an auxiliary variable $y_{v}$.
This can be done by the following constraints:
\begin{align*}
  y_{v} &\ge \textstyle 
   (\ell(v) + \sum_{t,p} \ell(v,t,p) \cdot x_{t,p} )
  -(r(v) + \sum_{t,p} r(v,t,p) \cdot x_{t,p} ), \\
  y_{v} &\ge \textstyle
   (r(v) + \sum_{t,p} r(v,t,p) \cdot x_{t,p} )
  -(\ell(v) + \sum_{t,p} \ell(v,t,p) \cdot x_{t,p} ).
\end{align*}
Then the imbalance of the whole ordering, which is our objective function to minimize, can be expressed as
\[
  \textstyle\sum_{v \in S} y_{v} + \sum_{t,p} \im(t,p) \cdot x_{t,p}.
\]
Now we need the following constraints to keep the total number of type-$t$ components right:
\[
  \textstyle\sum_{p} x_{t,p} = c_{t} \quad \text{for each type} \ t,
\]
where $c_{t}$ is the number of components of type $t$ in $G - S$.

By finding an optimal solution to the ILP above for each guess of the relative ordering of $S$,
we can find an optimal ordering.
Since the number of guesses and the number of variables depend only on $k$,
the theorem follows.
\end{proof}


\section{\textsc{Maximum Common (Induced) Subgraph}}
\label{sec:mcs}

In this section, we show that \textsc{Maximum Common Subgraph} (MCS) and \textsc{Maximum Common Induced Subgraph} (MCIS)
are fixed-parameter tractable parameterized by $\vi$ of both graphs.
(See Section~\ref{asec:mcs} for the proof for MCIS.)
The results extend known results and fill some complexity gaps as described below.

A graph $Q$ is \emph{subgraph-isomorphic} to $G$, denoted $Q \preceq G$,
if there is an injection $\eta$ from $V(Q)$ to $V(G)$
such that $\{\eta(u),\eta(v)\} \in E(G)$ for every $\{u,v\} \in E(Q)$.
A graph $Q$ is \emph{induced subgraph-isomorphic} to $G$, denoted $Q \indsub G$,
if there is an injection $\eta$ from $V(Q)$ to $V(G)$
such that $\{\eta(u),\eta(v)\} \in E(G)$ if and only if $\{u,v\} \in E(Q)$.
Given two graphs $G$ and $Q$,
\textsc{Subgraph Isomorphism} (SI) asks whether $Q \preceq G$,
and \textsc{Induced Subgraph Isomorphism} (ISI) asks whether $Q \indsub G$.
The results of this section are on their generalizations.
Given two graphs $G_{1}$ and $G_{2}$,
MCS asks to find a graph $H$ with maximum $|E(H)|$ such that $H \preceq G_{1}$ and $H \preceq G_{2}$.
Similarly,
MCIS asks to find a graph $H$ with maximum $|V(H)|$ such that $H \indsub G_{1}$ and $H \indsub G_{2}$.

If we restrict the structure of only one of the input graphs, then both problems remain quite hard.
Since \textsc{Partition Into Triangles}~\cite{GareyJ79} is a special case of SI
where the graph $Q$ is a disjoint union of triangles, 
MCS is NP-hard even if one of the input graphs has $\vi = 3$.
Also, since \textsc{Independent Set}~\cite{GareyJ79} is a special case of ISI where $Q$ is an edge-less graph,
MCIS is NP-hard even if one of the input graphs has $\vc = 0$.
Furthermore, since SI and ISI generalize \textsc{Clique}~\cite{DowneyF99},
MCS and MCIS are W[1]-hard parameterized by the order of one of the input graphs.
When parameterized by $\vc$ of one graph,
an XP algorithm for (a generalization of) MCS is known~\cite{BodlaenderHJOOZ20}.

For parameters restricting both input graphs, some partial results were known.
It is known that SI is fixed-parameter tractable parameterized by $\vi$ of both graphs,
while it is NP-complete when both graphs have $\td \le 3$~\cite{BodlaenderHOOZ19}.
The hardness proof in~\cite{BodlaenderHOOZ19} can be easily adapted to ISI without increasing $\td$.
It is known that MCIS is fixed-parameter tractable parameterized by $\vc$ of both graphs~\cite{Abu-Khzam14}.

\begin{theorem}
\label{thm:vi-mcs-both}
\textsc{Maximum Common Subgraph}
is fixed-parameter tractable parameterized by $\vi$ of both input graphs.
\end{theorem}
\begin{proof}
Let $G_{1} = (V_{1}, E_{1})$ and $G_{2} = (V_{2}, E_{2})$ be the input graphs of vertex integrity at most $k$.
We will find isomorphic subgraphs
$\Gamma_{1} = (U_{1}, F_{1})$ of $G_{1}$ and
$\Gamma_{2} = (U_{2}, F_{2})$ of $G_{2}$ with maximum number of edges,
and an isomorphism $\eta \colon U_{1} \to U_{2}$ from $\Gamma_{1}$ to $\Gamma_{2}$.

\smallskip
 
\noindent\textit{Step 1. Guessing matched $\vi(2k)$-sets $R_{1}$ and $R_{2}$.}
Let $S_{1}$ and $S_{2}$ be $\vi(k)$-sets of $G_{1}$ and $G_{2}$, respectively.
At this point, there is no guarantee that $S_{i} \subseteq U_{i}$ or $\eta(S_{1}) = S_{2}$.
To have such assumptions, we make some guesses about $\eta$
and find $\vi(2k)$-sets $R_{1}$ and $R_{2}$ of the graphs such that $\eta(R_{1}) = R_{2}$.

\textit{Step 1-1. Guessing subsets $X_{i}, Y_{i} \subseteq S_{i}$ for $i \in \{1,2\}$.}
We guess disjoint subsets $X_{1}$ and $Y_{1}$ of $S_{1}$ such that 
$X_{1} = S_{1} \cap {\eta^{-1}(U_{2} \cap S_{2})}$ and
$Y_{1} = S_{1} \cap {\eta^{-1}(U_{2} \setminus S_{2})}$.
We also  guess disjoint subsets $X_{2}$ and $Y_{2}$ of $S_{2}$
defined similarly as 
$X_{2} = S_{2} \cap {\eta(U_{1} \cap S_{1})}$ and
$Y_{2} = S_{2} \cap {\eta(U_{1} \setminus S_{1})}$.
Note that $\eta(X_{1}) = X_{2}$.
There are $3^{|S_{1}|} \cdot 3^{|S_{2}|} \le 3^{2k}$ candidates for the combinations of $X_{1}$, $Y_{1}$, $X_{2}$, and $Y_{2}$.

Observe that the vertices in $S_{i} \setminus (X_{i} \cup Y_{i})$ do not contribute to the isomorphic subgraphs
and can be safely removed. We denote the resultant graphs by $H_{i}$.

\smallskip

\textit{Step 1-2. Guessing $\eta$ on $X_{1} \cup Y_{1}$  and $\eta^{-1}$ on $X_{2} \cup Y_{2}$.}
Given the guessed subsets $X_{1}$, $Y_{1}$, $X_{2}$, and $Y_{2}$,
we further guess how $\eta$ maps these subsets.
There are $|X_{1}|! \le k!$ candidates for the bijection $\eta|_{X_{1}}$
(equivalently for $\eta^{-1}|_{X_{2}} = (\eta|_{X_{1}})^{-1}$).

Now we guess $\eta|_{Y_{1}}$ from at most $2^{k^{3}}$ non-isomorphic candidates as follows.
Recall that $\eta(Y_{1}) \subseteq V_{2} \setminus S_{2}$.
Observe that each subset $A \subseteq V_{2} \setminus S_{2}$ is completely characterized up to isomorphism
by the numbers of ways $A$ intersects type-$t$ components
for all $(H_{2},S_{2})$-types $t$.
Since there are at most $2^{\binom{k}{2}}$ types and each component has order at most $k$,
the total number of non-equivalent subsets of components is at most $2^{\binom{k}{2}} \cdot 2^{k} \le 2^{k^{2}}$.
Since $\eta(Y_{1})$ is the union of at most $|Y_{1}|$ such subsets,
the number of non-isomorphic candidates of $\eta(Y_{1})$ is at most $(2^{k^{2}})^{|Y_{1}|} \le 2^{k^{3}}$.
In the analogous way, we can guess $\eta^{-1}|_{Y_{2}}$ from at most $2^{k^{3}}$ non-isomorphic candidates.

Now we set $Z_{1} = \eta^{-1}(Y_{2})$ and $Z_{2} = \eta(Y_{1})$.
Let $R_{1} = X_{1} \cup Y_{1} \cup Z_{1}$ and $R_{2} = X_{2} \cup Y_{2} \cup Z_{2}$.
Observe that each component $C$ of $H_{1} - R_{1}$ satisfies that
$|C| \le k - |S_{1}| \le k$ and
$|C| + |R_{1}| \le (k - |S_{1}|) + (|S_{1}| + |\eta^{-1}(Y_{2})|) \le 2k$.
Hence, $R_{1}$ is a $\vi(2k)$-set of $H_{1}$.
Similarly, we can see that $R_{2}$ is a $\vi(2k)$-set of $H_{2}$.
Furthermore, we know that $\eta(R_{1}) = R_{2}$.

\medskip

\noindent\textit{Step 2. Extending the guessed parts of $\eta$.}
Assuming that the guesses we made so far are correct, we now find the entire $\eta$.
Recall that we are seeking for isomorphic subgraphs
$\Gamma_{1} = (U_{1}, F_{1})$ of $G_{1}$ and
$\Gamma_{2} = (U_{2}, F_{2})$ of $G_{2}$ with maximum number of edges,
and the isomorphism $\eta \colon U_{1} \to U_{2}$ from $\Gamma_{1}$ to $\Gamma_{2}$.
Since we already know the part $\eta|_{R_{1}} \colon R_{1} \to R_{2}$,
it suffices to find a bijective mapping from a subset of $V(H_{1} - R_{1})$ to a subset of $V(H_{2} - R_{1})$
that maximizes the number of matched edges where the connections to $R_{i}$ are also taken into account.

As we describe below, the subproblem we consider here can be solved by formulating it as an ILP instance with $2^{O(k^{3})}$ variables.
The trick here is that instead of directly finding the mapping,
we find which vertices and edges in $H_{i} - R_{i}$ are used in the common subgraph.

In the following, we are going to use a generalized version of \emph{types}
since the vertex set of a component of $H_{i} - R_{i}$ does not necessarily induce a connected subgraph of $\Gamma_{i}$.
It is defined in a similar way as $(H_{i}, R_{i})$-types
except that it is defined for each pair $(A,B)$ of a connected subgraph $A$ of $H_{i} - R_{i}$ and a subset $B$ of the edges between $A$ and $R_{i}$.
Let $(A_{1}, B_{1})$ and $(A_{2},B_{2})$ be such pairs in $H_{i} - R_{i}$.
We say that $(A_{1}, B_{1})$ and $(A_{2},B_{2})$ have the same \emph{g-$(H_{i}, R_{i})$-type} (or just \emph{g-type})
if there is an isomorphism from $H_{i}(A_{1}, B_{1})$ to $H_{i}(A_{2}, B_{2})$ that fixes $R_{i}$,
where $H_{i}(A_{j},B_{j})$ is the subgraph of $H_{i}$ formed by $B_{j}$ and the edges in $A_{j}$.
See \figref{fig:g-type}.
We say that a pair $(A,B)$ is of \emph{g-$(H_{i}, R_{i})$-type $t$} (or just \emph{g-type $t$})
by using a canonical form $t$ of the g-$(H_{i}, R_{i})$-type equivalence class of $(A,B)$.
Observe that all possible canonical forms of g-types can be computed in time depending only on $k$.

\begin{figure}[htb]
  \centering
  \includegraphics[scale=.8]{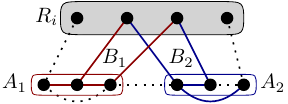} 
  \caption{The pairs $(A_{1}, B_{1})$ and $(A_{2},B_{2})$ have the same g-$(H_{i}, R_{i})$-type.}
  \label{fig:g-type}
\end{figure}

\smallskip

\textit{Step 2-1. Decomposing components of $H_{i} - R_{i}$ into smaller pieces.}
We say that an edge $\{u,v\}$ in $H_{1}$ is \emph{used by $\eta$}
if $u, v \in U_{1}$ and $H_{2}$ has the edge $\{\eta(u),\eta(v)\}$.
Similarly, an edge $\{u,v\}$ in $H_{2}$ is \emph{used by $\eta$}
if $u, v \in U_{2}$ and $H_{1}$ has the edge $\{\eta^{-1}(u),\eta^{-1}(v)\}$.

Let $i \in \{1,2\}$, $t$ be an $(H_{i}, R_{i})$-type, and $T$ be a multiset of g-$(H_{i}, R_{i})$-types.
Let $C$ be a type $t$ component of $H_{i} - R_{i}$,
$C'$ the subgraph of $C$ formed by the edges used by $\eta$, and 
$E'$ the subset of the edges between $C'$ and $R_{i}$ used by $\eta$.
If $T$ coincides with the multiset of g-types of the pairs $(A,B)$ such that $A$ is a component of $C'$
and $B$ is the subset of $E'$ connecting $A$ and $R_{i}$,
then we say that $\eta$ \emph{decomposes} the type-$t$ component $C$ into $T$.

We represent by a nonnegative variable $x^{(i)}_{t,T}$
the number of type-$t$ components of $H_{i} - R_{i}$
that are decomposed into $T$ by $\eta$.
We have the following constraint:
\begin{align}
  \textstyle
  \sum_{T} x^{(i)}_{t,T} = c^{(i)}_{t}
  \quad \text{for each} \ (H_{i}, R_{i})\text{-type} \ t \ \text{and} \ i \in \{1,2\},
  \nonumber 
\end{align}
where the sum is taken over all possible multisets $T$ of g-$(H_{i}, R_{i})$-types,
and $c^{(i)}_{t}$ is the number of components of type $t$ in $H_{i} - R_{i}$.
Additionally, if there is no way to decompose a type-$t$ component into $T$,
we add a constraint $x^{(i)}_{t,T} = 0$.

As each component of $H_{i} - R_{i}$ has order at most $k$, $T$ contains at most $k$ elements.
Since there are at most $2^{\binom{2k}{2}}$ g-types,
there are at most $(2^{\binom{2k}{2}})^{k}$ options for choosing $T$.
Thus the number of variables $x^{(i)}_{t,T}$ is at most 
$2 \cdot 2^{\binom{2k}{2}} \cdot (2^{\binom{2k}{2}})^{k+1}$.

Now we introduce a nonnegative variable $y_{t}^{(i)}$
that represents the number of pairs $(A,B)$ of g-type $t$ 
obtained from the components of $H_{i} - R_{i}$ by decomposing them by $\eta$.
The definition of $y_{t}^{(i)}$ gives the following constraint:
\begin{align}
  \textstyle
  y_{t}^{(i)} = \sum_{t', \, T} \mu(T, t) \cdot x^{(i)}_{t',T}
  \quad \text{for each} \  \text{g-}(H_{i}, R_{i})\text{-type} \ t \ \text{and} \ i \in \{1,2\},
  \nonumber 
\end{align}
where $\mu(T, t)$ is the multiplicity of $g$-type $t$ in $T$
and the sum is taken over all possible $(H_{i}, R_{i})$-types $t'$ and
multisets $T$ of g-$(H_{i}, R_{i})$-types.
As in the previous case, we can see that the number of variables $y_{t}$ depends only on $k$.

\smallskip

\textit{Step 2-2. Matching decomposed pieces.}
Observe that for each g-$(H_{1}, R_{1})$-type $t_{1}$,
there exists a unique g-$(H_{2}, R_{2})$-type $t_{2}$
such that 
there is an isomorphism $g$ from $H_{1}(A_{1}, B_{1})$ to $H_{2}(A_{2}, B_{2})$ with $g|_{R_{1}} = \eta|_{R_{1}}$,
where $(A_{i},B_{i})$ is a pair of g-$(H_{i}, R_{i})$-type $t_{i}$ for $i \in \{1,2\}$.
We say that such g-types $t_{1}$ and $t_{2}$ \emph{match}.
Since $\eta$ is an isomorphism from $\Gamma_{1}$ to $\Gamma_{2}$,
$\eta$ maps each g-$(H_{1}, R_{1})$-type $t_{1}$ pair 
to a g-$(H_{2}, R_{2})$-type $t_{2}$ pair, where $t_{1}$ and $t_{2}$ match.
This implies that $y_{t_{1}}^{(1)} = y_{t_{2}}^{(2)}$, which we add as a constraint.
Now the total number of edges used by $\eta$ can be computed from $y_{t}^{(1)}$.
Let $m_{t}$ be the number of edges in $H_{1}(A, B)$, where $(A,B)$ is a pair of g-$(H_{1}, R_{1})$-type $t$.
Let $r$ be the number of matched edges in $R_{1}$;
that is, $r = |\{\{u,v\} \in E(H_{1}[R_{1}]) \mid \{\eta(u), \eta(v)\} \in E(G_{2}[R_{2}]) \}|$.
Then, the number of matched edges is $r + \sum_{t} m_{t} \cdot y_{t}^{(1)}$.
On the other hand, given an assignment to the variables, 
it is easy to find isomorphic subgraphs with that many edges.
Since $r$ is a constant here,
we set $\sum_{t} m_{t} \cdot y_{t}^{(1)}$
to the objective function to be maximized.

Since the number of candidates in the guesses we made
and the number of variables in the ILP instances depend only on $k$,
the theorem follows.
\end{proof}


\section{\textsc{Min Max Outdegree Orientation}}
\label{sec:min-max-outdeg}

Given an undirected graph $G = (V,E)$, an edge weight function $w\colon E \to \mathbb{Z}^{+}$, and a positive integer $r$,
\textsc{Min Max Outdegree Orientation} (MMOO) asks whether 
there exists an orientation $\Lambda$ of $G$ such that each vertex has outdegree at most $r$ under $\Lambda$,
where the outdegree of a vertex is the sum of the weights of out-going edges.
If each edge weight is given in binary, we call the problem \textsc{Binary MMOO},
and if it is given in unary, we call the problem \textsc{Unary MMOO}.
Note that in the binary version, the weight of an edge can be exponential in the input size,
whereas the unary version does not allow such weights.

\textsc{Unary MMOO} 
admits an $n^{O(\tw)}$-time algorithm~\cite{Szeider11},
but it is W[1]-hard parameterized by $\td$~\cite{Szeider11arxiv}.\footnote{%
In \cite{Szeider11arxiv}, W[1]-hardness was stated for $\tw$ but the proof shows it for $\td$ as well.}
In this section, we show a stronger hardness parameterized by $\vc$. 
\textsc{Binary MMOO} is known to be NP-complete for graphs of $\vi = 4$~\cite{AsahiroMO11}.
In Section \ref{asec:min-max-outdeg}, we show a stronger hardness result that the binary version is NP-complete for graphs of $\vc = 3$.
This result is tight as we can show that the binary version is polynomial-time solvable for graphs of $\vc \le 2$.

\begin{theorem}
\label{thm:uminmaxod-vc}
\textsc{Unary MMOO} is W[1]-hard parameterized by $\vc$.
\end{theorem}
\begin{proof}
We give a parameterized reduction from \textsc{Unary Bin Packing}.
Given a positive integer $t$ and $n$ positive integers $a_{1}, a_{2}, \dots, a_{n}$ in unary,
\textsc{Unary Bin Packing} asks the existence of a partition $S_{1}, \dots, S_{t}$ of $\{1, 2, \dots, n\}$ 
such that $\sum_{i \in S_{j}} a_{i} = \frac{1}{t}\sum_{1 \le i \le n} a_{i}$ for $1 \le j \le t$.
\textsc{Unary Bin Packing} is W[1]-hard parameterized by $t$~\cite{JansenKMS13}.

We assume that $t \ge 3$ since otherwise the problem can be solved in polynomial time as the integers $a_{i}$ are given in unary.
Let $B = \frac{1}{t} \sum_{1 \le i \le n} a_{i}$ and $W = (t - 1)B = \sum_{1 \le i \le n} a_{i} - B$.
The assumption $t \ge 3$ implies that $B \le W/2$.
Observe that if $a_{i} \ge B$ for some $i$,
then the instance is a trivial no instance (when $a_{i} > B$)
or the element $a_{i}$ is irrelevant (when $a_{i} = B$).
Hence, we assume that $a_{i} < B$ (and thus $a_{i} < W/2$) for every $i$.

The reduction to \textsc{Unary MMOO} is depicted in \figref{fig:minmaxod-vc3}.
From the integers $a_{1}, a_{2}, \dots, a_{n}$,
we construct the graph obtained from 
a complete bipartite graph on the vertex set $\{u, s_{1}, s_{2}, \dots, s_{t}\} \cup \{v_{1}, \dots, v_{n}\}$
by adding the edge $\{u, s_{1}\}$.
We set $w(\{v_{i}, s_{j}\}) = a_{i}$ for all $i,j$,
$w(\{v_{i}, u\}) = W - a_{i}$ for all $i$,
and $w(\{u, s_{1}\}) = W$.
The vertices $s_{1}, s_{2}, \dots, s_{t}, u$ form a vertex cover of size $t + 1$.
We set the target maximum outdegree $r$ to $W$.
We show that this instance of \textsc{Unary MMOO} is a yes instance 
if and only if there exists a partition $S_{1}, \dots, S_{t}$ of $\{1, 2, \dots, n\}$ such that 
$\sum_{i \in S_{j}} a_{i} = B$ for all $j$.
Intuitively, we can translate the solutions of the problems 
by picking $a_{i}$ into $S_{j}$ if $\{v_{i}, s_{j}\}$ is oriented from $v_{i}$ to $s_{j}$, and vice versa.

Assume that there exists a partition $S_{1}, \dots, S_{t}$ of $\{1, 2, \dots, n\}$ such that 
$\sum_{i \in S_{j}} a_{i} = B$ for all $j$.
We first orient the edge $\{u,s_{1}\}$ from $u$ to $s_{1}$ and each edge $\{v_{i}, u\}$ from $v_{i}$ to $u$.
(See the thick edges in \figref{fig:minmaxod-vc3}.)
Then, we orient $\{v_{i}, s_{j}\}$ from $v_{i}$ to $s_{j}$ if and only if $i \in S_{j}$.
Under this orientation, all vertices have outdegree exactly $W$:
$a_{i} + (W-a_{i})$ for each $v_{i}$ and
$\sum_{i \notin S_{j}} a_{i} =  \sum_{1 \le i \le n} a_{i} - B$ for each $s_{j}$.

Conversely, assume that there is an orientation 
such that each vertex has outdegree at most $W$.
Since the sum of the edge weights is $(n+t+1)W$ and the graph has $n+t+1$ vertices,
the outdegree of each vertex has to be exactly $W$.
Since $a_{i} < W/2$ for all $i$, each edge $\{v_{i}, u\}$ has weight larger than $W/2$.
Hence,  for $u$, the only way to obtain outdegree exactly $W$ is 
to orient $\{u,s_{1}\}$ from $u$ to $s_{1}$
and $\{v_{i},u\}$ from $v_{i}$ to $u$ for all $i$.
Furthermore, for each $i$, there exists exactly one vertex $s_{j}$
such that $\{v_{i}, s_{j}\}$ is oriented from $v_{i}$ to $s_{j}$.
Let $S_{j} \subseteq \{1,2, \dots, n\}$ be the set of indices $i$ such that
$\{v_{i}, s_{j}\}$ is oriented from $v_{i}$ to $s_{j}$.
The discussion above implies that $S_{1}, \dots, S_{t}$ is a partition of $\{1,\dots,n\}$.
The outdegree of $s_{j}$ is $\sum_{i \notin S_{j}} a_{i}$, which is equal to $W = \sum_{1 \le i \le n} a_{i} - B$.
Thus, $\sum_{i \in S_{j}} a_{i} = \sum_{1 \le i \le n} a_{i} - W = B$.
\end{proof}
\begin{figure}[tb]
  \centering
  \includegraphics[scale=.8]{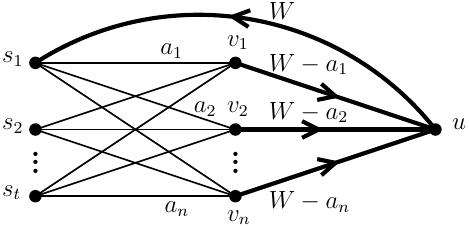} 
  \caption{Reduction from \textsc{Unary Bin Packing} to \textsc{Unary MMOO}.}
  \label{fig:minmaxod-vc3}
\end{figure}


\section{\textsc{Bandwidth}}
\label{sec:bandwidth}
Let $G = (V,E)$ be a graph.
Given a linear ordering $\sigma$ on $V$,
the \emph{stretch} of $\{u,v\} \in E$, denoted $\mathsf{str}_{\sigma}(\{u,v\})$, is $|\sigma(u) - \sigma(v)|$.
The \emph{bandwidth} of $G$, denoted $\bw(G)$,
is defined as $\min_{\sigma} \max_{e \in E} \mathsf{str}_{\sigma}(e)$,
where the minimum is taken over all linear orderings on $V$.
Given a graph $G$ and an integer $w$,
\textsc{Bandwidth} asks whether $\bw(G) \le w$.
\textsc{Bandwidth} is NP-complete on trees of $\pw = 3$~\cite{Monien86}
and on graphs of $\pw = 2$~\cite{Muradian03}.
Fellows et al.~\cite{FellowsLMRS08} presented an FPT algorithm for \textsc{Bandwidth} parameterized by $\vc$ .
Here we show that \textsc{Bandwidth} is W[1]-hard parameterized by $\td$ on trees.
The proof is inspired by the one by Muradian~\cite{Muradian03}. 

\begin{theorem}
\label{thm:bandwidth}
 \textsc{Bandwidth} is W[1]-hard parameterized by $\td$ on trees.
\end{theorem}
\begin{proof}
Let $(a_{1}, \dots, a_{n}; t)$ be an instance of \textsc{Unary Bin Packing} with $t \ge 2$.
Let $B = \frac{1}{t}\sum_{1 \le i \le n} a_{i}$ be the target weight.
We construct an equivalent instance $(T =(V,E), w)$ of \textsc{Bandwidth} as follows~(see \figref{fig:bw-td}).
We start with a path $(z_{0}, x_{1}, y_{1}, z_{1}, \dots, x_{t}, y_{t}, z_{t})$ of length $3t$.
For $1 \le i \le t-1$, we attach $12tnB$ leaves to $z_{i}$.
To $z_{0}$ and $z_{t}$, we attach $12tnB + 4n + 1$ leaves.
For $1 \le i \le n$, we take a star with $6tn \cdot a_{i}-1$ leaves centered at $v_{i}$.
Finally, we connect each $v_{i}$ to $x_{1}$ with a path with $6t-4$ inner vertices.
We set the target width $w$ to $6tnB + 2n + 1$.
Note that $|V| = (3t + 2)w +1$.

\begin{figure}[b]
  \centering
  \includegraphics[width=\textwidth]{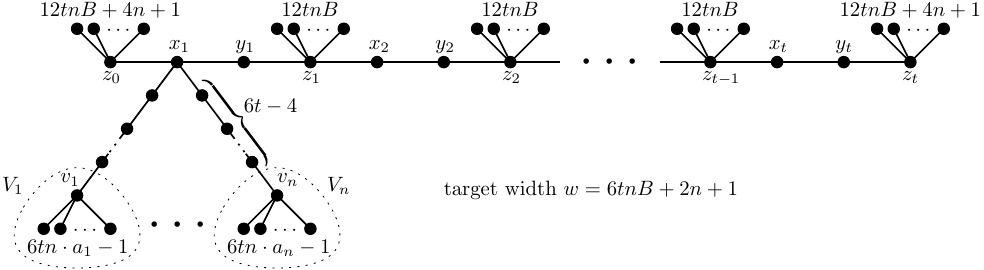} 
  \caption{Reductions from \textsc{Unary Bin Packing} to \textsc{Bandwidth}.}
  \label{fig:bw-td}
\end{figure}

We can see an upper bound of $\td(T)$ as follows.
We remove $x_{1}$ and all the leaves from $T$.
This decreases treedepth by at most $2$.
The remaining graph is a disjoint union of paths
and a longest path has order $6t-3$.
Since $\td(P_{n}) = \lceil \log_{2} (n+1) \rceil$~\cite{NesetrilO2012},
we have $\td(T) \le 2 + \lceil \log_{2} (6t-2) \rceil \le \log_{2}t + 6$.

Now we show that
$(T, w)$ is a yes instance of of \textsc{Bandwidth}
if and only if
$(a_{1}, \dots, a_{n}; t)$ is a yes instance of \textsc{Unary Bin Packing}.

\smallskip

($\implies$)
First assume that $\bw(T) \le w$ and 
that $\sigma$ is a linear ordering on $V$ such that $\max_{e \in E}\mathsf{str}_{\sigma}(e) \le w$.
Since $\deg(z_{0}) = 12tnB+4n+2 = 2w$,
its closed neighborhood $N[z_{0}]$ has to appear in $\sigma$ consecutively,
where $z_{0}$ appears at the middle of this subordering.
Furthermore, no edge can connect a vertex appearing before $z_{0}$ in $\sigma$
and a vertex appearing after $z_{0}$ as such an edge has stretch larger than $w$.
Since the edges not incident to $z_{0}$ form a connected subgraph,
we can conclude that the vertices in $V - N[z_{0}]$ appear 
either all before $N[z_{0}]$ or all after $N[z_{0}]$ in $\sigma$.
By symmetry, we can assume that those vertices appear after $N[z_{0}]$ in $\sigma$.
This implies that $\sigma(z_{0}) = w+1$.
By the same argument, we can show that all vertices in $N[z_{t}]$ appear consecutively in the end of $\sigma$
and $\sigma(z_{t}) = |V|-w = (3t+1)w+1$.
Since $\sigma(z_{t}) - \sigma(z_{0}) = 3tw$
and the path $(z_{0}, x_{1}, y_{1}, z_{1}, \dots, x_{t}, y_{t}, z_{t})$ has length $3t$,
each edge in this path has stretch exactly $w$ in $\sigma$.
Namely, $\sigma(x_{i}) = (3i-1)w + 1$, $\sigma(y_{i}) = 3iw + 1$, and $\sigma(z_{i}) = (3i+1)w + 1$.

For each leaf $\ell$ attached to $z_{i}$ ($1 \le i \le t-1$),
$\sigma(y_{i}) < \sigma(\ell) < \sigma(x_{i+1})$ holds.
Other than these leaves, there are $2(w-1) - 12tnB = 4n$ vertices placed between $y_{i}$ and $x_{i+1}$.
Let $V_{i}$ be the set consisting of $v_{i}$ and the leaves attached to it.
For $j \in \{1,\dots,t\}$, let $I_{j}$ be the set of indices $i$
such that $v_{i}$ is put between $z_{j-1}$ and $z_{j}$.
If $i \in I_{j}$, then all $6tn \cdot a_{i}$ vertices in $V_{i}$ are put between $y_{j-1}$ and $x_{j+1}$.
(We set $y_{0} \coloneqq z_{0}$.)

If $\sum_{i \in I_{j}} a_{i} \ge B+1$, then
$|\bigcup_{i \in I_{j}} V_{i}| \ge 6tn(B+1) > w+8n-1$ as $t \ge 2$.
This number of vertices cannot be put between $y_{j-1}$ and $x_{j+1}$
after putting the leaves attached to $z_{j-1}$ and $z_{j}$:
we can put at most $4n$ vertices between $y_{j-1}$ and $x_{j}$, at most $4n$ vertices between $y_{j}$ and $x_{j+1}$,
and at most $w-1$ vertices between $x_{j}$ and $y_{j}$.
Since $I_{1}, \dots, I_{t}$ form a partition of $\{1,\dots,n\}$
and $\sum_{1 \le i \le n} a_{i} = t B$,
we can conclude that $\sum_{i \in I_{j}} a_{i} = B$ for $1 \le j \le t$.

\smallskip

($\impliedby$)
Next assume that there exists a partition $S_{1}, \dots, S_{t}$ of $\{1,2,\dots,n\}$
such that $\sum_{i \in S_{j}} a_{i} = B$ for all $1 \le j \le t$.

We put $N[z_{0}]$ at the beginning of $\sigma$ and $N[z_{t}]$ at the end.
We set $\sigma(x_{i}) = (3i-1)w + 1$, $\sigma(y_{i}) = 3iw + 1$, and $\sigma(z_{i}) = (3i+1)w + 1$.
For $1 \le i \le t-1$, we put the leaves attached to $z_{i}$
so that a half of them have the first $6tnB$ positions between $y_{i}$ and $z_{i}$
and the other half have the first $6tnB$ positions between $z_{i}$ and $x_{i+1}$.
For each $S_{j}$, we put the vertices in $\bigcup_{i \in S_{j}} V_{i}$
so that they take the first $6tnB$ positions between $x_{j}$ and $y_{j}$.

Now we have $2n$ vacant positions at the end of each interval
between $x_{i}$ and $y_{i}$ for $1 \le i \le t$,
between $y_{i}$ and $z_{i}$ for $1 \le i \le t-1$, and
between $z_{i}$ and $x_{i+1}$ for $1 \le i \le t-1$.
To these positions, we need to put the inner vertices of the paths connecting $x_{1}$ and $v_{1}, \dots, v_{n}$.
Let $P_{i}$ be the inner part of $x_{1}$--$v_{i}$ path.
The path $P_{i}$ uses the $(2i-1)$st and $(2i)$th vacant positions in each interval as follows 
(see \figref{fig:bw-td_path}).

Let $i \in S_{j}$. 
Starting from $x_{1}$, $P_{i}$ proceeds from left to right and
visits the two positions in each interval consecutively
until it arrives the interval between $x_{j}$ and $y_{j}$.
At the interval between $x_{j}$ and $y_{j}$, $P_{i}$ switches to the phase where it 
only visits the $(2i)$th vacant position in each interval and still proceeds from left to right 
until it reaches the interval between $x_{t}$ and $y_{t}$.
Then $P_{i}$ changes the direction and switches to the phase where 
it visits the $(2i-1)$st vacant position only in each interval 
until it reaches the interval between $x_{j}$ and $y_{j}$.

Now all the vertices are put at distinct positions
and it is easy to see that no edge has stretch more than $w$.
This completes the proof.
\end{proof}
\begin{figure}[tb]
  \centering
  \includegraphics[width=\textwidth]{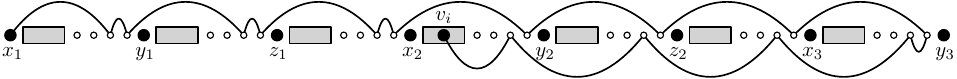} 
  \caption{Embedding the path from $x_{1}$ to $v_{i}$. 
    The gray boxes are the occupied position
    and the white points are the vacant positions.
    ($n = 2$, $j = 2$, $t = 3$.)}
  \label{fig:bw-td_path}
\end{figure}


\section{Conclusion}

Using vertex integrity as a structural graph parameter,
we presented finer analyses of the parameterized complexity of well-studied problems.
Although we needed a case-by-case analysis depending on individual problems,
the results in this paper would be useful for obtaining a general method to deal with vertex integrity.

Although we succeeded to extend many fixed-parameter algorithms parameterized by $\vc$ to the ones parameterized by $\vi$,
we were not so successful on graph layout problems.
Fellows et al.~\cite{FellowsLMRS08} showed that \textsc{Imbalance}, \textsc{Bandwidth}, \textsc{Cutwidth}, and \textsc{Distortion}
are fixed-parameter tractable parameterized by $\vc$.
Lokshtanov~\cite{Lokshtanov15} showed that \textsc{Optimal Linear Arrangement} is fixed-parameter tractable parameterized by $\vc$.
Are these problems fixed-parameter tractable parameterized by $\vi$?
We answered only for \textsc{Imbalance} in this paper.



\bibliography{vi}

\begin{thebibliography}{10}

\bibitem{Abu-Khzam14}
Faisal~N. Abu{-}Khzam.
\newblock Maximum common induced subgraph parameterized by vertex cover.
\newblock {\em Inf. Process. Lett.}, 114(3):99--103, 2014.
\newblock \href {https://doi.org/10.1016/j.ipl.2013.11.007}
  {\path{doi:10.1016/j.ipl.2013.11.007}}.

\bibitem{ArnborgLS91}
Stefan Arnborg, Jens Lagergren, and Detlef Seese.
\newblock Easy problems for tree-decomposable graphs.
\newblock {\em J. Algorithms}, 12(2):308--340, 1991.
\newblock \href {https://doi.org/10.1016/0196-6774(91)90006-K}
  {\path{doi:10.1016/0196-6774(91)90006-K}}.

\bibitem{AsahiroMO11}
Yuichi Asahiro, Eiji Miyano, and Hirotaka Ono.
\newblock Graph classes and the complexity of the graph orientation minimizing
  the maximum weighted outdegree.
\newblock {\em Discret. Appl. Math.}, 159(7):498--508, 2011.
\newblock \href {https://doi.org/10.1016/j.dam.2010.11.003}
  {\path{doi:10.1016/j.dam.2010.11.003}}.

\bibitem{AssmannPSZ81}
S.~F. Assmann, G.~W. Peck, M.~M. Sys{\l}o, and J.~Zak.
\newblock The bandwidth of caterpillars with hairs of length 1 and 2.
\newblock {\em SIAM Journal on Algebraic Discrete Methods}, 2(4):387--393,
  1981.
\newblock \href {https://doi.org/10.1137/0602041} {\path{doi:10.1137/0602041}}.

\bibitem{BarefootES87}
Curtis~A. Barefoot, Roger~C. Entringer, and Henda~C. Swart.
\newblock Vulnerability in graphs --- a comparative survey.
\newblock {\em J. Combin. Math. Combin. Comput.}, 1:13--22, 1987.

\bibitem{BelmonteFGR17}
R{\'{e}}my Belmonte, Fedor~V. Fomin, Petr~A. Golovach, and M.~S. Ramanujan.
\newblock Metric dimension of bounded tree-length graphs.
\newblock {\em {SIAM} J. Discret. Math.}, 31(2):1217--1243, 2017.
\newblock \href {https://doi.org/10.1137/16M1057383}
  {\path{doi:10.1137/16M1057383}}.

\bibitem{BelmonteHK0L18}
R{\'{e}}my Belmonte, Tesshu Hanaka, Ioannis Katsikarelis, Eun~Jung Kim, and
  Michael Lampis.
\newblock New results on directed edge dominating set.
\newblock In {\em {MFCS} 2018}, volume 117 of {\em LIPIcs}, pages 67:1--67:16,
  2018.
\newblock \href {https://doi.org/10.4230/LIPIcs.MFCS.2018.67}
  {\path{doi:10.4230/LIPIcs.MFCS.2018.67}}.

\bibitem{BelmonteKLMO20}
R{\'{e}}my Belmonte, Eun~Jung Kim, Michael Lampis, Valia Mitsou, and Yota
  Otachi.
\newblock Grundy distinguishes treewidth from pathwidth.
\newblock In {\em {ESA} 2020}, volume 173 of {\em LIPIcs}, pages 14:1--14:19,
  2020.
\newblock \href {https://doi.org/10.4230/LIPIcs.ESA.2020.14}
  {\path{doi:10.4230/LIPIcs.ESA.2020.14}}.

\bibitem{BiedlCGHW05}
Therese~C. Biedl, Timothy~M. Chan, Yashar Ganjali, Mohammad~Taghi Hajiaghayi,
  and David~R. Wood.
\newblock Balanced vertex-orderings of graphs.
\newblock {\em Discret. Appl. Math.}, 148(1):27--48, 2005.
\newblock \href {https://doi.org/10.1016/j.dam.2004.12.001}
  {\path{doi:10.1016/j.dam.2004.12.001}}.

\bibitem{Bodlaender88}
Hans~L. Bodlaender.
\newblock Dynamic programming on graphs with bounded treewidth.
\newblock In {\em ICALP 1988}, volume 317 of {\em Lecture Notes in Computer
  Science}, pages 105--118, 1988.
\newblock \href {https://doi.org/10.1007/3-540-19488-6_110}
  {\path{doi:10.1007/3-540-19488-6_110}}.

\bibitem{BodlaenderHJOOZ20}
Hans~L. Bodlaender, Tesshu Hanaka, Lars Jaffke, Hirotaka Ono, Yota Otachi, and
  Tom~C. van~der Zanden.
\newblock Hedonic seat arrangement problems (extended abstract).
\newblock In {\em {AAMAS} 2020}, pages 1777--1779, 2020.
\newblock URL: \url{https://dl.acm.org/doi/abs/10.5555/3398761.3398979}.

\bibitem{BodlaenderHOOZ19}
Hans~L. Bodlaender, Tesshu Hanaka, Yoshio Okamoto, Yota Otachi, and Tom~C.
  van~der Zanden.
\newblock Subgraph isomorphism on graph classes that exclude a substructure.
\newblock In {\em {CIAC} 2019}, volume 11485 of {\em Lecture Notes in Computer
  Science}, pages 87--98, 2019.
\newblock \href {https://doi.org/10.1007/978-3-030-17402-6_8}
  {\path{doi:10.1007/978-3-030-17402-6_8}}.

\bibitem{BonnetP19}
{\'{E}}douard Bonnet and Nidhi Purohit.
\newblock Metric dimension parameterized by treewidth.
\newblock In {\em {IPEC} 2019}, volume 148 of {\em LIPIcs}, pages 5:1--5:15,
  2019.
\newblock \href {https://doi.org/10.4230/LIPIcs.IPEC.2019.5}
  {\path{doi:10.4230/LIPIcs.IPEC.2019.5}}.

\bibitem{BonnetS17}
{\'{E}}douard Bonnet and Florian Sikora.
\newblock The graph motif problem parameterized by the structure of the input
  graph.
\newblock {\em Discret. Appl. Math.}, 231:78--94, 2017.
\newblock \href {https://doi.org/10.1016/j.dam.2016.11.016}
  {\path{doi:10.1016/j.dam.2016.11.016}}.

\bibitem{ChenKX10}
Jianer Chen, Iyad~A. Kanj, and Ge~Xia.
\newblock Improved upper bounds for vertex cover.
\newblock {\em Theor. Comput. Sci.}, 411(40-42):3736--3756, 2010.
\newblock \href {https://doi.org/10.1016/j.tcs.2010.06.026}
  {\path{doi:10.1016/j.tcs.2010.06.026}}.

\bibitem{CorneilR05}
Derek~G. Corneil and Udi Rotics.
\newblock On the relationship between clique-width and treewidth.
\newblock {\em {SIAM} J. Comput.}, 34(4):825--847, 2005.
\newblock \href {https://doi.org/10.1137/S0097539701385351}
  {\path{doi:10.1137/S0097539701385351}}.

\bibitem{Courcelle92}
Bruno Courcelle.
\newblock The monadic second-order logic of graphs {III:} tree-decompositions,
  minor and complexity issues.
\newblock {\em {RAIRO} Theor. Informatics Appl.}, 26:257--286, 1992.
\newblock \href {https://doi.org/10.1051/ita/1992260302571}
  {\path{doi:10.1051/ita/1992260302571}}.

\bibitem{CourcelleMR00}
Bruno Courcelle, Johann~A. Makowsky, and Udi Rotics.
\newblock Linear time solvable optimization problems on graphs of bounded
  clique-width.
\newblock {\em Theory Comput. Syst.}, 33(2):125--150, 2000.
\newblock \href {https://doi.org/10.1007/s002249910009}
  {\path{doi:10.1007/s002249910009}}.

\bibitem{CyganFKLMPPS15}
Marek Cygan, Fedor~V. Fomin, {\L}ukasz Kowalik, Daniel Lokshtanov, D{\'{a}}niel
  Marx, Marcin Pilipczuk, Micha{\l} Pilipczuk, and Saket Saurabh.
\newblock {\em Parameterized Algorithms}.
\newblock Springer, 2015.
\newblock \href {https://doi.org/10.1007/978-3-319-21275-3}
  {\path{doi:10.1007/978-3-319-21275-3}}.

\bibitem{DomLSV08}
Michael Dom, Daniel Lokshtanov, Saket Saurabh, and Yngve Villanger.
\newblock Capacitated domination and covering: {A} parameterized perspective.
\newblock In {\em {IWPEC} 2008}, volume 5018 of {\em Lecture Notes in Computer
  Science}, pages 78--90, 2008.
\newblock \href {https://doi.org/10.1007/978-3-540-79723-4_9}
  {\path{doi:10.1007/978-3-540-79723-4_9}}.

\bibitem{DowneyF99}
Rodney~G. Downey and Michael~R. Fellows.
\newblock {\em Parameterized Complexity}.
\newblock Springer, 1999.
\newblock \href {https://doi.org/10.1007/978-1-4612-0515-9}
  {\path{doi:10.1007/978-1-4612-0515-9}}.

\bibitem{DrangeDH16}
P{\aa}l~Gr{\o}n{\aa}s Drange, Markus~S. Dregi, and Pim van~'t Hof.
\newblock On the computational complexity of vertex integrity and component
  order connectivity.
\newblock {\em Algorithmica}, 76(4):1181--1202, 2016.
\newblock \href {https://doi.org/10.1007/s00453-016-0127-x}
  {\path{doi:10.1007/s00453-016-0127-x}}.

\bibitem{DvorakEGKO17}
Pavel Dvo\v{r}\'{a}k, Eduard Eiben, Robert Ganian, Du\v{s}an Knop, and
  Sebastian Ordyniak.
\newblock Solving integer linear programs with a small number of global
  variables and constraints.
\newblock In {\em {IJCAI} 2017}, pages 607--613, 2017.
\newblock \href {https://doi.org/10.24963/ijcai.2017/85}
  {\path{doi:10.24963/ijcai.2017/85}}.

\bibitem{DvorakK18}
Pavel Dvo\v{r}\'{a}k and Du\v{s}an Knop.
\newblock Parameterized complexity of length-bounded cuts and multicuts.
\newblock {\em Algorithmica}, 80(12):3597--3617, 2018.
\newblock \href {https://doi.org/10.1007/s00453-018-0408-7}
  {\path{doi:10.1007/s00453-018-0408-7}}.

\bibitem{EncisoFGKRS09}
Rosa Enciso, Michael~R. Fellows, Jiong Guo, Iyad~A. Kanj, Frances~A. Rosamond,
  and Ondrej Such{\'{y}}.
\newblock What makes equitable connected partition easy.
\newblock In {\em {IWPEC} 2009}, volume 5917 of {\em Lecture Notes in Computer
  Science}, pages 122--133, 2009.
\newblock \href {https://doi.org/10.1007/978-3-642-11269-0_10}
  {\path{doi:10.1007/978-3-642-11269-0_10}}.

\bibitem{FellowsFLRSST11}
Michael~R. Fellows, Fedor~V. Fomin, Daniel Lokshtanov, Frances~A. Rosamond,
  Saket Saurabh, Stefan Szeider, and Carsten Thomassen.
\newblock On the complexity of some colorful problems parameterized by
  treewidth.
\newblock {\em Inf. Comput.}, 209(2):143--153, 2011.
\newblock \href {https://doi.org/10.1016/j.ic.2010.11.026}
  {\path{doi:10.1016/j.ic.2010.11.026}}.

\bibitem{FellowsLMRS08}
Michael~R. Fellows, Daniel Lokshtanov, Neeldhara Misra, Frances~A. Rosamond,
  and Saket Saurabh.
\newblock Graph layout problems parameterized by vertex cover.
\newblock In {\em ISAAC 2008}, volume 5369 of {\em Lecture Notes in Computer
  Science}, pages 294--305, 2008.
\newblock \href {https://doi.org/10.1007/978-3-540-92182-0_28}
  {\path{doi:10.1007/978-3-540-92182-0_28}}.

\bibitem{FialaGK11}
Ji\v{r}{\'{\i}} Fiala, Petr~A. Golovach, and Jan Kratochv{\'{\i}}l.
\newblock Parameterized complexity of coloring problems: Treewidth versus
  vertex cover.
\newblock {\em Theor. Comput. Sci.}, 412(23):2513--2523, 2011.
\newblock \href {https://doi.org/10.1016/j.tcs.2010.10.043}
  {\path{doi:10.1016/j.tcs.2010.10.043}}.

\bibitem{FrankT87}
Andr{\'{a}}s Frank and {\'{E}}va Tardos.
\newblock An application of simultaneous diophantine approximation in
  combinatorial optimization.
\newblock {\em Combinatorica}, 7:49--65, 1987.
\newblock \href {https://doi.org/10.1007/BF02579200}
  {\path{doi:10.1007/BF02579200}}.

\bibitem{Gabow83}
Harold~N. Gabow.
\newblock An efficient reduction technique for degree-constrained subgraph and
  bidirected network flow problems.
\newblock In {\em STOC 1983}, pages 448--456, 1983.
\newblock \href {https://doi.org/10.1145/800061.808776}
  {\path{doi:10.1145/800061.808776}}.

\bibitem{Ganian11}
Robert Ganian.
\newblock Twin-cover: Beyond vertex cover in parameterized algorithmics.
\newblock In {\em {IPEC} 2011}, volume 7112 of {\em Lecture Notes in Computer
  Science}, pages 259--271, 2011.
\newblock \href {https://doi.org/10.1007/978-3-642-28050-4_21}
  {\path{doi:10.1007/978-3-642-28050-4_21}}.

\bibitem{Ganian12arxiv}
Robert Ganian.
\newblock Using neighborhood diversity to solve hard problems.
\newblock {\em CoRR}, abs/1201.3091, 2012.
\newblock \href {http://arxiv.org/abs/1201.3091} {\path{arXiv:1201.3091}}.

\bibitem{GanianKO18}
Robert Ganian, Fabian Klute, and Sebastian Ordyniak.
\newblock On structural parameterizations of the bounded-degree vertex deletion
  problem.
\newblock {\em Algorithmica}, 2020.
\newblock \href {https://doi.org/10.1007/s00453-020-00758-8}
  {\path{doi:10.1007/s00453-020-00758-8}}.

\bibitem{GareyJ79}
Michael~R. Garey and David~S. Johnson.
\newblock {\em Computers and Intractability: {A} Guide to the Theory of
  {NP}-Completeness}.
\newblock W. H. Freeman, 1979.

\bibitem{Gassner10}
Elisabeth Gassner.
\newblock The steiner forest problem revisited.
\newblock {\em J. Discrete Algorithms}, 8(2):154--163, 2010.
\newblock \href {https://doi.org/10.1016/j.jda.2009.05.002}
  {\path{doi:10.1016/j.jda.2009.05.002}}.

\bibitem{HlinenyOSG08}
Petr Hlinen{\'{y}}, Sang{-}il Oum, Detlef Seese, and Georg Gottlob.
\newblock Width parameters beyond tree-width and their applications.
\newblock {\em Comput. J.}, 51(3):326--362, 2008.
\newblock \href {https://doi.org/10.1093/comjnl/bxm052}
  {\path{doi:10.1093/comjnl/bxm052}}.

\bibitem{JansenKMS13}
Klaus Jansen, Stefan Kratsch, D{\'{a}}niel Marx, and Ildik{\'{o}} Schlotter.
\newblock Bin packing with fixed number of bins revisited.
\newblock {\em J. Comput. Syst. Sci.}, 79(1):39--49, 2013.
\newblock \href {https://doi.org/10.1016/j.jcss.2012.04.004}
  {\path{doi:10.1016/j.jcss.2012.04.004}}.

\bibitem{Kannan87}
Ravi Kannan.
\newblock Minkowski's convex body theorem and integer programming.
\newblock {\em Math. Oper. Res.}, 12:415--440, 1987.
\newblock \href {https://doi.org/10.1287/moor.12.3.415}
  {\path{doi:10.1287/moor.12.3.415}}.

\bibitem{KellerhalsK20}
Leon Kellerhals and Tomohiro Koana.
\newblock Parameterized complexity of geodetic set.
\newblock In {\em {IPEC} 2020}, volume 180 of {\em LIPIcs}, pages 20:1--20:14,
  2020.
\newblock \href {https://doi.org/10.4230/LIPIcs.IPEC.2020.20}
  {\path{doi:10.4230/LIPIcs.IPEC.2020.20}}.

\bibitem{Lenstra83}
Hendrik~W. {Lenstra Jr.}
\newblock Integer programming with a fixed number of variables.
\newblock {\em Math. Oper. Res.}, 8:538--548, 1983.
\newblock \href {https://doi.org/10.1287/moor.8.4.538}
  {\path{doi:10.1287/moor.8.4.538}}.

\bibitem{Lokshtanov15}
Daniel Lokshtanov.
\newblock Parameterized integer quadratic programming: Variables and
  coefficients.
\newblock {\em CoRR}, abs/1511.00310, 2015.
\newblock \href {http://arxiv.org/abs/1511.00310} {\path{arXiv:1511.00310}}.

\bibitem{LokshtanovMS13}
Daniel Lokshtanov, Neeldhara Misra, and Saket Saurabh.
\newblock Imbalance is fixed parameter tractable.
\newblock {\em Inf. Process. Lett.}, 113(19-21):714--718, 2013.
\newblock \href {https://doi.org/10.1016/j.ipl.2013.06.010}
  {\path{doi:10.1016/j.ipl.2013.06.010}}.

\bibitem{MeeksS16}
Kitty Meeks and Alexander Scott.
\newblock The parameterised complexity of list problems on graphs of bounded
  treewidth.
\newblock {\em Inf. Comput.}, 251:91--103, 2016.
\newblock \href {https://doi.org/10.1016/j.ic.2016.08.001}
  {\path{doi:10.1016/j.ic.2016.08.001}}.

\bibitem{MisraM20}
Neeldhara Misra and Harshil Mittal.
\newblock Imbalance parameterized by twin cover revisited.
\newblock In {\em {COCOON} 2020}, volume 12273 of {\em Lecture Notes in
  Computer Science}, pages 162--173, 2020.
\newblock \href {https://doi.org/10.1007/978-3-030-58150-3_13}
  {\path{doi:10.1007/978-3-030-58150-3_13}}.

\bibitem{Monien86}
Burkhard Monien.
\newblock The bandwidth minimization problem for caterpillars with hair length
  3 is {NP}-complete.
\newblock {\em {SIAM} J. Algebraic and Discrete Methods}, 7(4):505--512, 1986.
\newblock \href {https://doi.org/10.1137/0607057} {\path{doi:10.1137/0607057}}.

\bibitem{Muradian03}
David Muradian.
\newblock The bandwidth minimization problem for cyclic caterpillars with hair
  length 1 is {NP}-complete.
\newblock {\em Theor. Comput. Sci.}, 307(3):567--572, 2003.
\newblock \href {https://doi.org/10.1016/S0304-3975(03)00238-X}
  {\path{doi:10.1016/S0304-3975(03)00238-X}}.

\bibitem{NesetrilO2012}
Jaroslav Ne\v{s}et\v{r}il and Patrice~Ossona de~Mendez.
\newblock {\em Sparsity: Graphs, Structures, and Algorithms}.
\newblock Algorithms and combinatorics. Springer, 2012.
\newblock \href {https://doi.org/10.1007/978-3-642-27875-4}
  {\path{doi:10.1007/978-3-642-27875-4}}.

\bibitem{Oum08}
Sang{-}il Oum.
\newblock Approximating rank-width and clique-width quickly.
\newblock {\em {ACM} Trans. Algorithms}, 5(1):10:1--10:20, 2008.
\newblock \href {https://doi.org/10.1145/1435375.1435385}
  {\path{doi:10.1145/1435375.1435385}}.

\bibitem{ReidlRVS14}
Felix Reidl, Peter Rossmanith, Fernando~S{\'{a}}nchez Villaamil, and Somnath
  Sikdar.
\newblock A faster parameterized algorithm for treedepth.
\newblock In {\em {ICALP} 2014}, volume 8572 of {\em Lecture Notes in Computer
  Science}, pages 931--942.
\newblock \href {https://doi.org/10.1007/978-3-662-43948-7_77}
  {\path{doi:10.1007/978-3-662-43948-7_77}}.

\bibitem{Szeider11arxiv}
Stefan Szeider.
\newblock Not so easy problems for tree decomposable graphs.
\newblock {\em Ramanujan Mathematical Society, Lecture Notes Series},
  No.~13:179--190, 2010.
\newblock \href {http://arxiv.org/abs/1107.1177} {\path{arXiv:1107.1177}}.

\bibitem{Szeider11}
Stefan Szeider.
\newblock Monadic second order logic on graphs with local cardinality
  constraints.
\newblock {\em {ACM} Trans. Comput. Log.}, 12(2):12:1--12:21, 2011.
\newblock \href {https://doi.org/10.1145/1877714.1877718}
  {\path{doi:10.1145/1877714.1877718}}.

\end{thebibliography}


\appendix


\section{Graph parameters}
\label{sec:graph-parameters}

We give formal definitions of $\vc(G)$, $\td(G)$, and $\vi(G)$ only.
See~\cite{HlinenyOSG08} for the definitions of $\pw(G)$, $\tw(G)$, and clique-width $\cw(G)$.

\subsection{Vertex cover}
Let $G = (V,E)$ be a graph. 
A set $S \subseteq V$ is a \emph{vertex cover} of $G$
if each component of $G-S$ has exactly one vertex.
The \emph{vertex cover number} of $G$, denoted $\vc(G)$,
is the size of a minimum vertex cover of $G$.
It is known that a vertex cover of size $k$ (if exists)
can be found in time $O(2^{k} \cdot n)$~\cite{CyganFKLMPPS15}, where $n = |V|$
(see \cite{ChenKX10} for the currently fastest algorithm).
Thus we can assume that a vertex cover of minimum size is given
when designing an algorithm parameterized by $\vc$.

\subsection{Treedepth}

The treedepth of a graph $G = (V,E)$, denoted $\td(G)$, is defined recursively as follows:
\[
  \td(G) =
  \begin{cases}
    1 & |V| = 1,\\
    \max_{1 \le i \le c} \td(C_{i}) & G \textrm{ has $c\ge 2$ components } C_{1}, \dots, C_{c}, \\
    1 + \min_{v \in V} \td(G - v) & \textrm{otherwise}.
  \end{cases}
\]
In other words, a graph $G = (V,E)$ has treedepth at most $d$ if 
there is a rooted forest $F$ of height at most $d$ on the same vertex set $V$
such that two vertices are adjacent in $G$ only if one is an ancestor of the other in $F$.
It is known that such $F$, if exists, can be found in time $2^{O(d^{2})} \cdot n$~\cite{ReidlRVS14},
where $n = |V|$.
So we assume that such a rooted forest of depth $\td(G)$ is given together with $G$
when the parameter is $\td$.

From the rooted forest $F$, one can easily construct a path decomposition of $G$ with maximum bag size at most $d$:
use leaves as bags and put all ancestors of a leaf into the bag corresponding to the leaf.
This implies that $\pw(G) + 1 \le \td(G)$ for every graph $G$.
On the other hand, $\td$ cannot be bounded by any function of $\pw$ in general.
For example, $\pw(P_{n}) = 1$ and $\td(P_{n}) = \lceil \log_{2} (n+1) \rceil$~\cite{NesetrilO2012},
where $P_{n}$ is the path of order $n$.

In general, we have the following upper bound of the length of paths.
\begin{proposition}
[\cite{NesetrilO2012}]
\label{prop:td-path-length}
The length of a longest path in $G$ is less than $2^{\td(G)}$.
\end{proposition}

\subsection{Vertex integrity}
\label{sec:vi}
The \emph{vertex integrity}~\cite{BarefootES87} of a graph $G$, denoted $\vi(G)$, 
is the minimum integer $k$ satisfying that
there is a vertex set $S \subseteq V(G)$ such that 
$|S| + |V(C)| \le k$ for each component $C$ of $G-S$.
We call such $S$ a \emph{$\vi(k)$-set} of $G$.
For an $n$-vertex graph of vertex integrity at most $k$,
we can find a \emph{$\vi(k)$-set} in $O(k^{k+1} n)$ time~\cite{DrangeDH16}.
Hence, without loss of generality, we can assume that a $\vi(k)$-set is given as a part of input
when designing an FPT algorithm parameterized by $\vi$.
Observe that $\td(G) \le \vi(G)$ since we can first remove the $k' \le k$ vertices in a $\vi(k)$-set
and then each component has order at most $k-k'$ and thus treedepth at most $k-k'$.
Also, since a vertex cover of size $k$ is a $\vi(k+1)$-set, $\vi(G) \le \vc(G) + 1$ holds.

Dvo\v{r}\'{a}k et al.~\cite{DvorakEGKO17} showed that \textsc{Integer Linear Programming} (ILP)
is fixed-parameter tractable parameterized by the fracture number of the incidence, 
which is basically equivalent to the vertex integrity.
Ganian et al.~\cite{GanianKO18} showed that 
\textsc{Bounded Degree Deletion} is W[1]-hard parameterized by treedepth
but fixed-parameter tractable 
parameterized by core fracture number, which can be seen as a generalization of vertex integrity.
Bodlaender et al.~\cite{BodlaenderHOOZ19} showed that \textsc{Subgraph Isomorphism}
is fixed-parameter tractable parameterized by the vertex integrity of both graphs,
while the problem is NP-complete for graphs of treedepth $3$.


\section{ILP parameterized by the number of variables}
\label{sec:ilp}
Lenstra~\cite{Lenstra83} showed that the feasibility of an integer linear programming (ILP) formula can be decided in FPT time
when parameterized by the number of variables. The time and space complexity was later improved by Kannan~\cite{Kannan87}
and by Frank and Tardos~\cite{FrankT87}.
Their algorithms can be used also for the following ILP optimization problem (see e.g., \cite{FellowsLMRS08}).
\begin{myproblem}
  \problemtitle{\textsc{$p$-Opt-ILP}}
  \probleminput{A matrix $A \in \mathbb{Z}^{m \times p}$, vectors $b \in \mathbb{Z}^{m}$ and $c \in \mathbb{Z}^{p}$.}
  \problemquestiontitle{Task}
  \problemquestion{Find a vector $x \in \mathbb{Z}^{p}$ that minimizes $c^{\top} x$ and satisfies that $A x \ge b$.}
\end{myproblem}
\begin{proposition}
[\cite{Lenstra83,Kannan87,FrankT87}]
\label{prop:ILP_opt}
\textsc{$p$-Opt-ILP} can be solved using $O(p^{2.5 p + o(p)} \cdot L \cdot \log(MN))$ arithmetic operations
and space polynomial in $L$, where $L$ is the number of bits in the input,
$N$ is the maximum absolute value any variable can take,
and $M$ is an upper bound on the absolute value of the minimum taken by the objective function.
\end{proposition}


\section{Omitted proofs in Section~\ref{sec:mcs}}
\label{asec:mcs}

\begin{theorem}
\label{thm:vi-mcis-both}
\textsc{Maximum Common Induced Subgraph}
is fixed-parameter tractable parameterized by 
the sum of the vertex integrity of input graphs.
\end{theorem}
\begin{proof}
Since the proof is almost the same with the one for Theorem~\ref{thm:vi-mcs-both},
here we only describe the differences for handling induced subgraphs.

Let $G_{1} = (V_{1}, E_{1})$ and $G_{2} = (V_{2}, E_{2})$ be the input graphs of vertex integrity at most $k$.
We will find $U_{1} \subseteq V_{1}$ and $U_{2} \subseteq V_{2}$ with maximum size $|U_{1}| = |U_{2}|$
such that there is an isomorphism $\eta$ from $G_{1}[U_{1}]$ to $G_{2}[U_{2}]$.

\paragraph*{Step 1. Guessing matched $\vi(2k)$-sets $R_{1}$ and $R_{2}$.}
In the same way as before, we guess $\vi(2k)$-sets $R_{1}$ and $R_{2}$ of $G_{1}$ and $G_{2}$, respectively,
and a bijection $\eta|_{R_{1}}\colon R_{1} \to R_{2}$.
The only difference here is that we reject the current guess if 
$\eta|_{R_{1}}$ is not an isomorphism from $G_{1}[R_{1}]$ to $G_{2}[R_{2}]$.

\paragraph*{Step 2. Extending the guessed parts of $\eta$.}
To handle induced subgraphs, 
we need to modify the definition of 
``to decompose'' as follows:
for a type-$t$ component $C$ of $H_{i} - R_{i}$ and a multiset $T$ of g-$(H_{i}, R_{i})$-types, 
we say that $\eta\colon U_{1} \to U_{2}$ \emph{decomposes} $C$ into $T$
if $T$ coincides with the multiset of g-$(H_{i}, R_{i})$-types
of the pairs $(A,B)$ such that $A$ is a component of $H_{i}[V(C) \cap U_{i}]$
and $B$ is the set of all edges connecting $A$ and $R_{i}$.
Everything else works as before.
\end{proof}


\section{Omitted proofs in Section~\ref{sec:min-max-outdeg}}
\label{asec:min-max-outdeg}
\begin{theorem}
\label{thm:bminmaxod-vc3}
\textsc{Binary MMOO} is NP-complete
for graphs of $\vc = 3$.
\end{theorem}
\begin{proof}
Since the problem clearly belongs to NP, we show the NP-hardness 
by presenting a reduction from \textsc{Partition}, which is NP-complete~\cite{GareyJ79}.
Given an even number of positive integers $a_{1}, a_{2}, \dots, a_{n}$ in binary,
\textsc{Partition} asks the existence of a partition $\{S_{1}, S_{2}\}$ of $\{1, 2, \dots, n\}$ 
such that $\sum_{i \in S_{j}} a_{i} = \frac{1}{2}\sum_{1 \le i \le n} a_{i}$ for $j \in \{1,2\}$.
This problem remains NP-hard with an additional condition $|S_{1}| = |S_{2}| = n/2$~\cite{GareyJ79}.
We assume that $n \ge 10$ since otherwise the problem can be solved in polynomial time.
Let $B = W = \frac{1}{2}\sum_{1 \le i \le n} a_{i}$.

The proof is almost the same with the one of Theorem~\ref{thm:uminmaxod-vc} except that $t = 2$.
Observe that we assumed there that $t \ge 3$ only to guarantee that $a_{i} < W/2$ for all $i$.
To have this assumption here, we start with an instance $a_{1}, a_{2}, \dots, a_{n}$ 
of \textsc{Partition} with the restriction $|S_{1}| = |S_{2}| = n/2$.
Let $a_{i}' = a_{i} + B$ for each $i$,
and $B' = W' = \frac{1}{2}\sum_{1 \le i \le n} a_{i}' = (n/2+1) B$.
Clearly, this is an equivalent instance as we added the same value to each number.
Also, $a_{i}' < W'/2$ holds for all $i$ since $n \ge 10$ and $a_{i} < 2B$
imply that $a_{i}' = a_{i} + B < 3B \le (n/2+1) B / 2 = W'/2$.
Now we observe that the restriction $|S_{1}| = |S_{2}| = n/2$ is not a restriction anymore.
That is, if $\sum_{i \in S} a_{i}' = B'$ for some $S \subseteq \{1,\dots,n\}$, then $|S| = n/2$ holds.
Suppose to the contrary that $S \ne n/2$.
By swapping $S$ and $\{1,\dots,n\} \setminus S$ if necessary,
we can assume that $S \le n/2 - 1$.
This gives $(n/2+1) B = \sum_{i \in S} a_{i}' \le (n/2 - 1)B + \sum_{i \in S} a_{i}$,
which implies $\sum_{i \in S} a_{i} \ge 2B = \sum_{1 \le i \le n} a_{i}$, a contradiction.

We construct an instance of \textsc{Binary MMOO}
as exactly we did in the proof of Theorem~\ref{thm:uminmaxod-vc}
by setting $t = 2$ and using $a_{i}'$, $B'$, and $W'$ instead of $a_{i}$, $B$, and $W$.
The equivalence of the instances can be shown in the same way.
\end{proof}

\begin{theorem}
\label{thm:bminmaxod-vc2}
\textsc{Binary MMOO} can be solved in polynomial time for graphs of $\vc \le 2$.
\end{theorem}
\begin{proof}
Let $G = (V,E)$, $w \colon E \to \mathbb{Z}^{+}$, $r \in \mathbb{Z}^{+}$
be an instance of \textsc{Binary MMOO}.
We assume that $w(e) \le r$ for each $e \in E$ since otherwise the problem is trivial.
If there is a vertex of degree at most 1, we can safely remove it from the graph
since we can always orient the edge incident to the vertex (if exists) from the vertex to the other endpoint.
Hence, we assume that $G$ has minimum degree at least 2.

Let $\{p,q\} \subseteq V$ be a vertex cover of $G$.
By the assumption on the minimum degree, 
every vertex $v \in V \setminus \{p,q\}$ is adjacent to both $p$ and $q$.
If $w(\{v,p\}) + w(\{v,q\}) \le r$, then we can safely orient the edges from $v$ to $p$ and $q$.
Thus we remove such vertices from the graph.
Now it holds that $w(\{v,p\}) + w(\{v,q\}) > r$ for all $v \in V \setminus \{p,q\}$.
In particular, $\max\{w(\{v,p\}), w(\{v,q\})\} > r/2$ for all $v \in V \setminus \{p,q\}$.

Observe that for each vertex, at most one edge of weight more than $r/2$ can be oriented from the vertex to
one of its neighbors.
For $p$ and $q$, we guess such edges.
That is, we guess one edge of weight more than $r/2$ incident to $p$ ($q$, resp.)
and orient it from $p$ ($q$, resp.) to the other endpoint; or guess that there is no such edge.
These guesses determine almost a complete orientation.
For a non-guessed edge $\{v,p\}$ with $w(\{v,p\}) > r/2$, we orient it from $v$ to $p$.
Since $w(\{v,p\}) + w(\{v,q\}) > r$, we then have to orient $\{v,q\}$ from $q$ to $v$.
The other case of $w(\{v,q\}) > r/2$ is symmetric.
Now the only edge with undetermined orientation is $\{p,q\}$ (if they are adjacent).
We just try both directions of $\{p,q\}$ and check if the whole orientation is of maximum outdegree at most $r$.
\end{proof}

\section{Extending algorithms known for $\vc$ parameterizations}
\label{sec:extending-vc}

\subsection{Capacitated problems}
\label{sec:capacitated}

Let $G = (V,E)$ be a graph with a capacity function $c \colon V \to \mathbb{Z}^{+}$
such that $c(v) \le \deg(v)$ for each $v \in V$.
A set $C \subseteq V$ is a \emph{capacitated vertex cover}
if there exists a mapping $f \colon E \to C$ such that $f(e)$ is an endpoint of $e$ for each $e \in E$
and $|\{e \in E \mid f(e) = v\}| \le c(v)$ for each $v \in C$.
A set $D \subseteq V$ is a \emph{capacitated dominating set}
if there exists a mapping  $f \colon V \setminus D \to D$ such that $f(v) \in N(v) \cap D$ for each $v \in V \setminus D$
and $|\{v \in V \setminus D \mid f(v) = u\}| \le c(u)$ for each $u \in D$.
Now the problems studied in this section are defined as follows.
\begin{myproblem}
  \problemtitle{\textsc{Capacitated Vertex Cover}}
  \probleminput{A graph $G$, a capacity function $c \colon V \to \mathbb{Z}^{+}$, a positive integer $k$.}
  \problemquestiontitle{Question}
  \problemquestion{Is there a capacitated vertex cover $X$ of $G$ with $|X| \le k$?}
\end{myproblem}
\begin{myproblem}
  \problemtitle{\textsc{Capacitated Dominating Set}}
  \probleminput{A graph $G$, a capacity function $c \colon V \to \mathbb{Z}^{+}$, a positive integer $k$.}
  \problemquestiontitle{Question}
  \problemquestion{Is there a capacitated dominating set $D$ of $G$ with $|D| \le k$?}
\end{myproblem}
It is known that
\textsc{Capacitated Vertex Cover} is W[1]-hard parameterized by $\td$,
and \textsc{Capacitated Dominating Set} is W[1]-hard parameterized by $\td + k$~\cite{DomLSV08}.\footnote{%
The W[1]-hardness results are stated only for $\tw$ and $\tw + k$ but the proofs actually show them for
$\td$ and $\td+k$, respectively.}

For a vertex set $S$ of $G$,
we say that components $C_{1}$ and $C_{2}$ of $G - S$ have the same \emph{$c$-type}
if $C_{1}$ and $C_{2}$ have the same ($G, S$)-type and furthermore
there is an isomorphism $g$ from $G[S \cup V(C_{1})]$ to $G[S \cup V(C_{2})]$
such that $g|_{S}$ is the identity and $c(v) = c(g(v))$ for each $v \in S \cup V(C_{1})$.
We say that a component $C$ of $G - S$ is of \emph{$c$-type $t$} 
by using a canonical form $t$ of the members of the $c$-type equivalence class of $C$.
If $S$ is a $\vi(k)$-set of $G$, then every vertex in $G-S$ has degree less than $k$ in $G$,
and thus its capacity is also less than $k$.
This implies that the number of different $c$-types depends only on $k$.

\begin{theorem}
\label{thm:cvc}
 \textsc{Capacitated Vertex Cover} is fixed-parameter tractable parameterized by $\vi$.
\end{theorem}
\begin{proof}
We are going to find a minimum capacitated vertex cover $X$ of $G$.
Let $S$ be a $\vi(k)$-set of the input graph $G = (V,E)$.
We first guess the subset $X_{S} = X \cap S$
and the partial mapping $f_{S} \colon E(G[S]) \to X_{S}$ with $f_{S}(e) \in e$ for each $e \in E(G[S])$.
The numbers of candidates for $X_{S}$ and $f_{S}$ depend only on $k$.
For each $v \in X_{S}$, we set $c'(v) = c(v) - \left|\{e \in E(G[S]) \mid f_{S}(e) = v\}\right|$.
Each $v \in X_{S}$ can cover $c'(v)$ edges between $S$ and $V-S$.

Let $C$ be a $c$-type $t$ component of $G - S$.
We say that a pair $(W,f)$ of a subset $W \subseteq V(C)$ and a mapping $f \colon E(C) \cup E(V(C), S) \to W \cup X_{S}$
is \emph{feasible} if $|\{e \mid f(e) = v\}| \le c(v)$ for each $v \in W$
and $f(e) \in e$ for each $e \in E(C) \cup E(C, S)$.\footnote{%
For vertex sets $A$ and $B$, $E(A,B)$ denotes the set of edges between $A$ and $B$.}
The number of feasible pairs depends only on $k$.
A feasible pair gives a cover of all edges in $C$ and some edges between $V(C)$ and $S$,
and it asks $X_{S}$ to cover the remaining edges between $V(C)$ and $S$ in a certain way.
Now it suffices to find an assignment of feasible pairs to components of $G-S$
that minimizes the number of vertices used by the feasible pairs
and does not exceed the capacity of any vertex in $X_{S}$.

We represent by a nonnegative variable $x_{t,W,f}$ the number of $c$-type $t$ components $C$ of $G - S$
such that $V(C) \cap X = W$ and $(W,f)$ is a feasible pair.
The number of such variables depends only on $k$.
Let $d_{t}$ be the number of components of $c$-type $t$ in $G - S$.
Since each component of $G-S$ has to be assigned a feasible pair,
we have the following constraints:
\begin{align*}
  \textstyle
  \sum_{W, \; f} x_{t,W,f} = d_{t}
  \quad \text{for each} \ c\text{-type} \ t.
\end{align*}
The capacity constraints for $X_{S}$ can be expressed as follows:
\begin{align*}
  \textstyle
  \sum_{t, \; W, \; f} 
  \#(f,v) \cdot x_{t,W,f}  \le c'(v)
  \quad \text{for each} \ v \in X_{S},
\end{align*}
where $\#(f,v)$ is the number of edges that $f$ maps to $v$.
Finally, our objective function to minimize is $|X_{S}| + \sum_{t, \; W, \; f} |W| \cdot x_{t,W,f}$.

By finding an optimal solution to the ILP above for each guess of $X_{S}$ and $f_{S}$,
we can find the minimum capacitated vertex cover of $G$.
Since the number of guesses and the number of variables depend only on $k$,
the theorem follows by Proposition~\ref{prop:ILP_opt}.
\end{proof}

\begin{theorem}
\label{thm:cds}
\textsc{Capacitated Dominating Set} is fixed-parameter tractable parameterized by $\vi$
\end{theorem}
\begin{proof}
We are going to find a minimum capacitated dominating set $D$ of $G$.
Let $S$ be a $\vi(k)$-set of the input graph $G = (V,E)$.

We first guess the partition $(D_{S}, A_{S}, B_{S})$ of $S$
such that $D_{S} = D \cap S$, 
$A_{S}$ is the set of vertices dominated by $D_{S}$, and 
$B_{S}$ is the set of vertices dominated by $D \setminus S$.
Next we guess 
the partial mapping $f_{S} \colon A_{S} \to D_{S}$ with $f_{S}(v) \in N(v) \cap D_{S}$ for each $v \in A_{S}$.
The numbers of candidates for $(D_{S}, A_{S}, B_{S})$ and $f_{S}$ depend only on $k$.
For each $v \in D_{S}$, we set $c'(v) = c(v) - \left|\{u \in A_{S} \mid f_{S}(u) = v\}\right|$.
Each $v \in D_{S}$ can dominate $c'(v)$ vertices in $V-S$.

Let $C$ be a $c$-type $t$ component of $G - S$.
Let $(D_{C}, A_{C}, B_{C})$ be a partition of $V(C)$,
$B_{S}' \subseteq B_{S}$,
$f \colon A_{C} \cup B_{S}' \to D_{C}$,
and $g \colon B_{C} \to D_{S}$.
We say that $(D_{C}, A_{C}, B_{C}, B_{S}', f, g)$ is \emph{feasible} if 
$f(v) \in N(v) \cap D_{C}$ for each $v \in A_{C} \cup B_{S}'$,
$g(v) \in N(v) \cap D_{S}$ for each $v \in B_{C}$,
and $|\{u \in A_{C} \cup B_{S}' \mid f(u) = v\}| \le c(v)$ for each $v \in D_{C}$.
The number of feasible tuples depends only on $k$.
A feasible tuple gives a domination of all vertices in $V(C) \setminus B_{C}$ and $B_{S}'$,
and it asks $D_{S}$ to dominate $B_{C}$ in a certain way.

As before, it suffices to find an assignment of feasible tuples to components of $G-S$
that minimizes the number of vertices used by the feasible tuples
and does not exceed the capacity of any vertex in $D_{S}$.

We represent by a nonnegative variable $x_{t,D_{C},A_{C},B_{C},B_{S}',f,g}$
 the number of $c$-type $t$ components $C$ of $G - S$
such that $V(C) \cap D = D_{C}$ and $(D_{C},A_{C},B_{C},B_{S}',f,g)$ is a feasible tuple.
The number of such variables depends only on $k$.
Let $d_{t}$ be the number of components of $c$-type $t$ in $G - S$.
Since each component of $G-S$ has to be assigned a feasible tuple,
we have the following constraints:
\begin{align*}
  \sum_{D_{C}, \; A_{C}, \; B_{C}, \; B_{S}', \; f, \; g}
  x_{t,D_{C},A_{C},B_{C},B_{S}',f,g} = d_{t}
  \quad \text{for each} \ c\text{-type} \ t.
\end{align*}
The capacity constraints for $D_{S}$ can be expressed as follows:
\begin{align*}
  \sum_{D_{C}, \; A_{C}, \; B_{C}, \; B_{S}', \; f, \; g}
  \#(g,v) \cdot x_{t,D_{C},A_{C},B_{C},B_{S}',f,g} \le c'(v)
  \quad \text{for each} \ v \in D_{S},
\end{align*}
where $\#(g,v)$ is the number of vertices that $g$ maps to $v$.
We also have to guarantee that each vertex in $B_{S}$ is dominated by a vertex in $V - S$.
This can be done by the following constraints:
\begin{align*}
  \sum_{D_{C}, \; A_{C}, \; B_{C}, \; B_{S}' \ni v, \; f, \; g}
  x_{t,D_{C},A_{C},B_{C},B_{S}',f,g} \ge 1
  \quad \text{for each} \ v \in B_{S}.
\end{align*}
Finally, our objective function to minimize is
\[
  |D_{S}| + \sum_{D_{C}, \; A_{C}, \; B_{C}, \; B_{S}', \; f, \; g} |D_{C}| \cdot x_{t,D_{C},A_{C},B_{C},B_{S}',f,g}.
\]
As before the discussion so far implies the theorem.
\end{proof}


\subsection{Coloring and partitioning problems}
\label{sec:coloring}

\textsc{Precoloring Extension}, \textsc{Equitable Coloring}, and \textsc{Equitable Connected Partition}
form a first set of problems studied under the ``treewidth versus vertex cover'' 
perspective~\cite{EncisoFGKRS09,FellowsFLRSST11,FialaGK11}.
\textsc{Equitable Coloring} and \textsc{Precoloring Extension}
are fixed-parameter tractable parameterized by $\vc$~\cite{FialaGK11}
and W[1]-hard parameterized by $\td$~\cite{FellowsFLRSST11}.\footnote{%
The W[1]-hardness results are stated only for $\tw$ but the proofs actually show them for $\td$.}
\textsc{Equitable Connected Partition}
is fixed-parameter tractable parameterized by $\vc$
and W[1]-hard parameterized by $\pw$~\cite{EncisoFGKRS09}.

In this section, we show that all the three problems are fixed-parameter tractable parameterized by $\vi$.

\subsubsection{\textsc{Precoloring Extension}}
Given a graph $G = (V,E)$, a precoloring $c_{U}\colon U \to \{1,\dots,r\}$ for some $U \subseteq V$, and a positive integer $r$,
\textsc{Precoloring Extension} asks whether $G$ admits a proper $r$-coloring $c$ such that $c(v) = c_{U}(v)$ for every $v \in U$.

\begin{theorem}
\label{thm:prece}
\textsc{Precoloring Extension} is fixed-parameter tractable parameterized by $\vi$.
\end{theorem}
\begin{proof}
Let $(G=(V,E), c_{U}, r)$ be an instance of \textsc{Precoloring Extension}.
Let $S$ be a $\vi(k)$-set of $G$.
For each $v \in V$, let $L(v)$ be the following set (the list of allowed colors):
\[
  L(v) = 
  \begin{cases}
    \{c_{U}(v)\} & v \in U, \\
    \{1,\dots,r\} \setminus \{c_{U}(u) \mid u \in N(v) \cap U\} & v \in S \setminus U, \\
    \{1,\dots,\min\{r,k\}\} \setminus \{c_{U}(u) \mid u \in N(v) \cap U\} & v \in V \setminus (S \cup U).
  \end{cases}
\]
Observe that there exists a proper $r$-coloring $c$ of $G$ with $c(v) = c_{U}(v)$ for all $v \in U$
if and only if there is a proper coloring $c'$ of $G$ with $c'(v) \in L(v)$.
This is almost trivial except for the case of $v \in V \setminus (S \cup U)$,
where we restrict the domain to $\{1,\dots,k\}$ when $k < r$.
This can be justified by considering the degree of $v$.
Since $S$ is a $\vi(k)$-set and $v \notin S$, we have $\deg(v) < k$.
Thus, after coloring $G-v$, $v$ can always be colored with a color not used in its neighborhood.

Now in the list coloring setting, we can remove the vertices in $U$ unless 
the instance is a trivial no instance with $\{u,v\} \in E$ such that $c_{U}(u) = c_{U}(v)$.
In the following, we consider the graph where $U$ is removed
and still use the same symbols $G$ and $S$.

Let $v$ be a vertex with $|L(v)| \ge 2k$. By the definition of $L$, $v \in S$.
Such a vertex can be safely removed:
the vertices in $V \setminus S$ use colors only in $\{1,\dots,k\}$;
and the vertices in $S - v$ use at most $k-1$ colors in $L(v)$.
We now assume that $L(u) < 2k$ for all vertices in the graph.

Now we guess the coloring of $S$
and then check independently for each component $C$ of $G-S$ whether $G[S \cup V(C)]$ 
has a coloring consistent with $L$ and the guessed coloring of $S$.
The number of possible colorings of $S$ is at most $(2k)^{k}$.
Since $|S \cup V(C)| \le k$, checking the existence of a consistent coloring
can be done in time depending only on $k$.
\end{proof}

\subsubsection{\textsc{Equitable Coloring}}

Given an $n$-vertex graph $G = (V,E)$ and a positive integer $r$,
\textsc{Equitable Coloring} asks whether $G$ admits a proper $r$-coloring $c$ 
such that $|\{v \in V \mid c(v) = i\}| \in \{\lfloor n/r \rfloor, \lceil n/r \rceil\}$ 
for each $i \in \{1,\dots,r\}$.
We call such a coloring an \emph{equitable $r$-coloring}.

\begin{theorem}
\label{thm:eqcolor}
\textsc{Equitable Coloring} is fixed-parameter tractable parameterized by $\vi$.
\end{theorem}
\begin{proof}
Let $(G = (V,E), r)$ be an instance of \textsc{Equitable Coloring},
and $S$ be a $\vi(k)$-set of $G$.
We split the proof into two cases: $r \le 2k$ and $r > 2k$.
We reduce both cases to the feasibility test of the ILP defined as follows.
By Proposition~\ref{prop:ILP_opt}, the theorem will follow.

\paragraph*{Case 1: $r \le 2k$.}
We guess a partition $S_{1},\dots,S_{r}$ of $S$
such that each of them is an independent set and some of them may be empty.
Since $|S| \le k$ and $r \le 2k$, the number of such partitions depends only on $k$.
We interpret this partition as a coloring of $G[S]$ and try to extend this to the whole graph.

For a $(G,S)$-type $t$,
a coloring $\mu \colon V(C) \to \{1,\dots,r\}$ of a type-$t$ component $C$ of $G-S$ is \emph{feasible}
if $S_{i} \cup \{v \in V(C) \mid \mu(v) = i\}$ is an independent set for each $i$.
We set $\mu_{i} = |\{v \in V(C) \mid \mu(v) = i\}|$.

We represent by a nonnegative variable $x_{t,\mu}$ 
the number of type-$t$ components colored with a feasible $\mu$.
Since each component of $G-S$ has to be colored,
we have the following constraints:
\begin{align*}
  \textstyle
  \sum_{\mu}
  x_{t,\mu} = d_{t}
  \quad \text{for each} \ (G,S)\text{-type} \ t,
\end{align*}
where $d_{t}$ is the number of type-$t$ components in $G-S$.
The equitable constraints can be expressed as follows:
\begin{align*}
  \textstyle
  \sum_{t,\mu}
  \mu_{i} \cdot x_{t,\mu} &= \lceil n/r \rceil - |S_{i}|
  \quad \text{for each} \ i \in \{1,\dots,b\},
  \\
  \textstyle
  \sum_{t,\mu}
  \mu_{i} \cdot x_{t,\mu} &= \lfloor n/r \rfloor - |S_{i}|
  \quad \text{for each} \ i \in \{b+1,\dots,r\},
\end{align*}
where $b$ is the remainder of $n/r$.

\paragraph*{Case 2: $r > 2k$.}
In this case, we do not have an upper bound of $r$.
The first trick is that we can still guess the coloring of $S$
since we use at most $k$ colors there.
The second trick is that after checking the extendability of the $k$ colors,
the rest of the problem becomes trivial.

We guess a partition $S_{1},\dots,S_{k}$ of $S$ and an integer $a$ such that:
each $S_{i}$ is a possibly-empty independent set; and
there are disjoint independent sets $W_{1}, \dots, W_{k}$
such that $S_{i} \subseteq W_{i}$ for all $i$,
$|W_{i}| = \lceil n/r \rceil$ for $1 \le i \le a$,
and $|W_{i}| = \lfloor n/r \rfloor$ for $a+1 \le i \le k$.
Then, $G' \coloneqq G - \bigcup_{1 \le i \le k}W_{i}$ has an equitable $r-k$ coloring 
if and only if $G$ has an equitable $r$ coloring having $W_{1}, \dots, W_{k}$ as color classes.

We can show that actually $G'$ always has an equitable $r-k$ coloring.
Observe that each component in $G'$ has order at most $k < r-k$ as $G'$ is a subgraph of $G - S$.
We now linearly order the vertices of $G'$ in such a way that the vertices of a component appear consecutively.
Then we color the first vertex in this ordering with color $1$,
the second one with color $2$, and so on.
Formally, we color the $i$th vertex in this ordering with color $(i \bmod (r-k))+1$.
Since each component has order less than $r-k$, we never repeat a color in a component.
Thus, this is an equitable $r-k$ coloring of $G'$.
Therefore, it suffices to decide whether there exists the super sets $W_{1}, \dots, W_{k}$.

Since we are searching for a partial coloring, we use a special character $\ast$ to indicate ``not colored.''
We need to change the definition of feasibility.
For a $(G,S)$-type $t$,
a coloring $\mu \colon V(C) \to \{\ast\} \cup \{1,\dots,k\}$ of a type-$t$ component $C$ of $G-S$ is \emph{feasible},
if $S_{i} \cup \{v \in V(C) \mid \mu(v) = i\}$ is an independent set for each $i \ne \ast$.
We set $\mu_{i} = |\{v \in V(C) \mid \mu(v) = i\}|$. Now the rest of the proof is exactly the same as before.

We represent by a nonnegative variable $x_{t,\mu}$ 
the number of type-$t$ components colored with a feasible $\mu$.
Since each component of $G-S$ has to be colored,
we have the following constraints:
$\sum_{\mu} x_{t,\mu} = d_{t}$ for each $(G,S)$-type $t$,
where $d_{t}$ is the number of type-$t$ components in $G-S$.
The equitable constraints can be expressed as follows:
$\sum_{t,\mu} \mu_{i} \cdot x_{t,\mu} = \lceil n/r \rceil - |S_{i}|$ for $1 \le i \le a$, and
$\sum_{t,\mu} \mu_{i} \cdot x_{t,\mu} = \lfloor n/r \rfloor - |S_{i}|$ for $a+1 \le i \le k$.
\end{proof}

\subsubsection{\textsc{Equitable Connected Partition}}

Given an $n$-vertex graph $G = (V,E)$ and a positive integer $r$,
\textsc{Equitable Connected Partition} asks whether 
there is a partition of $V_{1},\dots,V_{r}$ of $V$
such that $G[V_{i}]$ is connected 
and $|V_{i}| \in \{\lfloor n/r \rfloor, \lceil n/r \rceil\}$ for all $i$.
We call such a partition an \emph{equitable connected $r$-partition}.

\begin{theorem}
\label{thm:eqconpar}
\textsc{Equitable Connected Partition} is fixed-parameter tractable parameterized by $\vi$.
\end{theorem}
\begin{proof}
Let $(G = (V,E), r)$ be an instance of \textsc{Equitable Connected Partition},
and $S$ be a $\vi(k)$-set of $G$.
Observe that at most $k$ of $V_{1}, \dots, V_{r}$ can intersect $S$.
We split the proof into two cases $r \le k$ and $r > k$.

\paragraph*{Case 1: $r \le k$.}
If additionally $\lfloor n/r \rfloor \le k$ holds in this case,
then $n \in O(k^{2})$. Thus we assume that $\lfloor n/r \rfloor > k$.
This implies that every $V_{i}$ intersects $S$.
We first guess the partition $S_{1}, \dots, S_{r}$ of $S$.

Let $C_{1}$ and $C_{2}$ be components of $G-S$ with the same type,
and $\mu_{j} \colon V(C_{j}) \to \{1,\dots,r\}$ for each $j \in \{1,2\}$.
Then, we say that $(C_{1},\mu_{1})$ and $(C_{2},\mu_{2})$ are \emph{equivalent}
if there is an isomorphism $\eta$ from $G[S \cup C_{1}]$ and $G[S \cup C_{2}]$
such that $\eta$ fixes $S$ and $\mu_{1}(v) = \mu_{2}(\eta(v))$ for all $v \in V(C_{1})$.
A set $\mathcal{M} = \{(C_{1}, \mu_{1}), \dots, (C_{p}, \mu_{p})\}$ is \emph{feasible} if 
$C_{j}$ is a component of $G-S$ for each $j$,
$\mu_{j} \colon V(C_{j}) \to \{1,\dots,r\}$ for each $j$,
and
the subgraph of $G$ induced by $S_{i} \cup \bigcup_{1 \le j \le p} \{v \in V(C_{j}) \mid \mu_{j}(v) = i\}$
is connected for all $1 \le i \le r$.
Let $\mathcal{M}'= \{(C_{1}', \mu_{1}'), \dots, (C_{q}', \mu_{q}')\}$ be a subset of $\mathcal{M}$
obtained by removing all but one of each equivalent class.
It is easy to see that $\mathcal{M}'$ is feasible if and only if so is $\mathcal{M}$.
Let $t_{j}'$ be the $(G,S)$-type of $C_{j}'$.
We call the set $\{(t_{1}', \mu_{1}'), \dots, (t_{q}', \mu_{q}')\}$
a \emph{type-color representation} of $\mathcal{M}$.

Now we guess the type-color representation $\mathcal{T} = \{(t_{1}, \mu_{1}), \dots, (t_{q}, \mu_{q})\}$ 
of a solution. That is, we find a partition such that at least one component of type $t_{1}$
is partitioned by $\mu_{1}$,
and no component is partitioned in a way not included in $\mathcal{T}$.
The number of candidates depends only on $k$,
and the feasibility of each candidate can be checked in polynomial time.

By a nonnegative variable $x_{t,\mu}$ for $(t,\mu) \in \mathcal{T}$,
we represent the number of type-$t$ components that we partition by $\mu$ or an equivalent mapping.
Since we take at least one such partition of type-$t$ components,
we set the constraint $x_{t,\mu} \ge 1$ for each $(t,\mu) \in \mathcal{T}$.
Now the connectivity has been handled, and we only need to force the equitable partition.
Since each component has to be partitioned, we need the following constraints:
\[
  \textstyle
  \sum_{(t,\mu) \in \mathcal{T}} x_{t,\mu} = d_{t} \quad \text{for each type} \ t,
\]
where $d_{t}$ is the number of type-$t$ components in $G - S$.
The equitable constraints can be expressed as follows:
\begin{align*}
  \textstyle
  \sum_{(t,\mu) \in \mathcal{T}}
  \mu^{(i)} \cdot x_{t,\mu} &= \lceil n/r \rceil - |S_{i}|
  \quad \text{for each} \ i \in \{1,\dots,a\},
  \\
  \textstyle
  \sum_{(t,\mu) \in \mathcal{T}}
  \mu^{(i)} \cdot x_{t,\mu} &= \lfloor n/r \rfloor - |S_{i}|
  \quad \text{for each} \ i \in \{a+1,\dots,r\},
\end{align*}
where $a$ is the remainder of $n/r$
and $\mu^{(i)}$ is the number of vertices $\mu$ maps to $i$.

Since the number of variables depends only on $k$,
Proposition~\ref{prop:ILP_opt} implies that the feasibility test of the ILP defined above
is fixed-parameter tractable parameterized by $k$.
 
\paragraph*{Case 2: $r > k$.}
In this case, some $V_{i}$ does not intersect $S$, and thus it is contained in a component of $G-S$.
This implies that $\lfloor n/r \rfloor \le k$, and thus $\max_{i} |V_{i}| \le k+1$.
We first guess the number $k' < k$ of the $V_{i}$'s intersecting $S$
and the number $a \le k'$ of size $\lceil n/r \rceil$ sets among them.
Now we guess $V_{1}$: guess at most $k+1$ types of $G-S$;
guess the number of components we take from the chosen types, which is at most $k+1$;
and for each component, guess the subset of the vertices taken to $V_{1}$.
The number of candidates depends only on $k$.
In general, when we guess $V_{i}$, $2 \le i \le k'$, 
we first remove the vertices chosen for $\bigcup_{1 \le j \le i-1} V_{j}$
and recompute and redefine the types. Then, we can guess $V_{i}$ in exactly the same way as the case of $i = 1$.
The number of candidates for all $V_{1}, \dots, V_{k'}$ depends only on $k$.
We reject the guess if some $G[V_{i}]$ is disconnected.

Let $W = \bigcup_{1 \le j \le k'} V_{j}$.
Now it suffices to decide whether $G - W$ has an equitable connected $(r-k')$-partition.
For each component $C$ of $G - W$,
we enumerate all the possible pairs $(p,q)$ of nonnegative integers
such that $C$ admits an equitable connected $(p+q)$-partition such that
$p$ parts have size $\lceil n/r \rceil$ and $q$ parts have size $\lfloor n/r \rfloor$.
This can be done in FPT time parameterized by $k$ in total,
since each component has at most $k$ vertices
and the number of components in $G-W$ is at most $|V|$.
We now check whether by picking one pair $(p,q)$ for each component,
it is possible to make the total number of components $r-k'$.
This can be done in polynomial time by a standard dynamic programming algorithm
since the number of components and $r-k'$ are at most $|V|$.
\end{proof}


\section{Hard problems parameterized by $\vi$}
\label{sec:hard-vi}

\subsection{\textsc{Graph Motif}}
Given a graph $G = (V,E)$, a vertex coloring $c \colon V \to \mathcal{C}$,
and a multiset $M$ of colors in $\mathcal{C}$,
the problem \textsc{Graph Motif} is to decide
if there is a set $S \subseteq V$ such that $G[S]$ is connected 
and $c(S) = M$, where $c(S)$ is the multiset of colors appearing in $S$.
If the motif $M$ is a set (i.e., no element in $M$ has multiplicity more than $1$),
then the restricted problem is called \textsc{Colorful Graph Motif}.
It is known that \textsc{Graph Motif} is fixed-parameter tractable parameterized by $\vc$~\cite{BonnetS17}
(actually by more general parameters \emph{neighborhood diversity}~\cite{Ganian12arxiv}
and \emph{twin-cover number}~\cite{Ganian11}).
The proof of Theorem~20 in~\cite{BonnetS17} implies that \textsc{Colorful Graph Motif} is NP-complete for graphs of $\vi = 6$.
By a similar proof, we will show that \textsc{Colorful Graph Motif} is NP-complete for graphs of $\vi = 4$.
We then complement this by showing that \textsc{Graph Motif} is polynomial-time solvable for graphs of $\vi \le 3$.

\begin{theorem}
\label{thm:mtfvi4}
\textsc{Colorful Graph Motif} is NP-complete on trees of vertex integrity~$4$.
\end{theorem}
\begin{proof}
The problem is clearly in NP\@.
We present a reduction from an NP-complete problem \textsc{3-Dimensional Matching}~\cite{GareyJ79}.
The input of \textsc{3-Dimensional Matching} consists of three disjoint sets
$X = \{x_{1}, \dots, x_{n}\}$, $Y = \{y_{1}, \dots, y_{n}\}$, $Z = \{z_{1}, \dots, z_{n}\}$,
and a set of triples $T \subseteq X \times Y \times Z$.
The task is to decide whether there is a subset $S$ of $T$
such that $|S| = n$ and each element of $X \cup Y \cup Z$ appears in a triple included in $S$.

We construct a graph $G$ with a coloring $c$ as follows.
The graph $G$ contains a special root vertex $r$ with unique color $c(r) = r$.
For each triple $t = (x_{i}, y_{j}, z_{k}) \in T$, take three new vertices
$t_{i}, t_{j}, t_{k}$ with $c(t_{i}) = x_{i}$, $c(t_{j}) = y_{j}$, and $c(t_{k}) = z_{k}$,
and add three new edges $\{r, t_{i}\}$, $\{t_{i}, t_{j}\}$, and $\{t_{j}, t_{k}\}$.
We set $M = \{r\} \cup X \cup Y \cup Z$. This completes the construction.
Note that $G$ is a tree and $\{r\}$ is a $\vi(4)$-set of $G$ 
as each component of $G - \{r\}$ is a path of order~$3$.

Assume that $(X,Y,Z,T)$ is a yes instance of \textsc{3-Dimensional Matching} with $S \subseteq T$ as a certificate.
We set $L = \{r\} \cup \{t_{i}, t_{j}, t_{k} \mid t = (x_{i}, y_{j}, z_{k}) \in S\}$.
Clearly, $c(L) = M$. Since $G[L]$ is connected, $(G,M)$ is a yes instance of \textsc{Colorful Graph Motif}.

To show the other direction, assume that a vertex subset $L$ of $G$ induces a connected graph and $c(L) = M$.
Observe that $L$ has to include $r$. 
Since $X \cup Y \cup Z = c(L) \setminus \{r\}$,
$L$ includes exactly $n$ vertices of distance $i$ from $r$ for each $i \in \{1,2,3\}$.
This fact and the connectivity of $G[L]$ imply that for each $t = (x_{i}, y_{j}, z_{k}) \in T$,
$L$ contains either all vertices $t_{i}, t_{j}, t_{k}$ or none of them.
Let $S \subseteq T$ be the set of triples such that $L$ contains all three vertices corresponding to each $t \in S$.
By the discussion above, $|S| = n$ and each element of $X \cup Y \cup Z$ appears in a triple included in $S$.
\end{proof}

\begin{theorem}
\label{thm:mtfvi3}
\textsc{Graph Motif} can be solved in polynomial time
on graphs of vertex integrity at most $3$.
\end{theorem}
\begin{proof}
Let $G = (V,E)$ be the input graph with a coloring $c \colon V \to \mathcal{C}$
and $M$ be the input multiset of colors.
Let $R$ be a $\vi(3)$-set of $G$.
If $|R| \ge 2$, then $R$ is a vertex cover of $G$ with $|R| \le 3$,
and thus we can apply an FPT algorithm parameterized by the vertex cover number~\cite{Ganian11,BonnetS17}.
If $R = \emptyset$, then each connected component of $G$ is of order at most $3$, and thus the problem is trivial.
In the following we assume that $R = \{r\}$ for some $r \in V$.
Furthermore, we assume that $r$ is included in the solution $S$ as otherwise $|S| \le 2$.
Let $D_{i}$ be the vertices of distance $i$ from $r$. Note that $V = \{r\} \cup D_{1} \cup D_{2}$.

We construct an auxiliary bipartite multi-graph $H$ as follows.
For each color $x \in \mathcal{C}$, take new vertices $x_{1}$ and $x_{2}$.
For each component $C$ of $G - r$,
if $C$ has two vertices $u$ of color $x$ and $v$ of color $y$,
where only $u$ is adjacent to $r$,
then add one edge between $x_{1}$ and $y_{2}$.
For each color $x \in \mathcal{C}$, we define degree constraints of $x_{1}$ and $x_{2}$ in $H$ as follows:
the degree constraint of $x_{1}$ is ``at most $M(x)$'',
and the degree constraint of $x_{2}$ is ``exactly $\max\{M(x)-q(x),0\}$'',
where $M(x)$ is the multiplicity of $x$ in $M \setminus \{c(r)\}$ and $q(x)$ is the number of color-$x$ vertices in $D_{1}$.
We will show that $H$ has a subgraph $F$ satisfying the degree constraints
if and only if 
there is a set $S \subseteq V$ such that $G[S]$ is connected and $c(S) = M$.
Since finding a subgraph of such degree constraints can be done in polynomial time~\cite{Gabow83},
this equivalence implies the theorem.

First assume that there is a set $S \subseteq V$ such that $G[S]$ is connected and $c(S) = M$.
We choose $S$ among such sets so that $|S \cap D_{1}|$ is maximized.
This implies in particular that if there is a vertex $v \in D_{1} \setminus S$ of color $x$,
then no vertex of color $x$ in $D_{2}$ belongs to $S$.
For each edge $\{u,v\}$ in $G[S-r]$, if $u \in D_{1}$, $v \in D_{2}$, $c(u) = x$, and $c(v) = y$,
then add one edge between $x_{1}$ and $y_{2}$ into $F$.
Now for each color $x$, the degree of $x_{1}$ in $F$ is at most $M(x)$.
Since $S$ takes color-$x$ vertices in $D_{2}$ only when it is necessary after
including all color-$x$ vertices in $D_{1}$,
the degree of $x_{2}$ in $F$ is exactly $\max\{M(x)-q(x),0\}$.

Next assume that $H$ has a subgraph $F$ that satisfies the degree constraints.
For each edge between $x_{1}$ and $y_{2}$ in $F$,
we add into $S$ the endpoints of an arbitrary edge $\{u,v\}$ in $G-r$
such that $c(u) = x$, $u \in D_{1}$, $c(v) = y$, and $v \in D_{2}$.
Let $\mathcal{S} = c(S)$, the multiset of colors appear in $S$.
From the construction of $S$, it holds for each color $x$
that $\mathcal{S}(x) = \deg_{F}(x_{1}) + \deg_{F}(x_{2}) = \deg_{F}(x_{1}) + \max\{M(x)-q(x),0\}$.
If $M(x) \le q(x)$, then $\mathcal{S}(x) = \deg_{F}(x_{1}) \le M(x)$.
We add, into $S$, arbitrary $M(x) - \mathcal{S}(x)$ of color-$x$ vertices in $D_{1} \setminus S$.
This is possible since the number of color-$x$ vertices in $D_{1} \setminus S$
is $q(x) - \mathcal{S}(x) \ge M(x) - \mathcal{S}(x)$.
If $M(x) > q(x)$, then $\mathcal{S}(x) = \deg_{F}(x_{1}) + M(x)-q(x)$.
In this case, we add all color-$x$ vertices in $D_{1}$ into $S$,
and then the multiplicity of $x$ in the resultant set becomes $q(x) + M(x)-q(x) = M(x)$.
\end{proof}

\subsection{\textsc{Steiner Forest}}
\textsc{Steiner Forest} is a generalization of \textsc{Steiner Tree} and defined as follows:
Given a graph $G = (V,E)$ with edge weighting $w \colon E \to \mathbb{Z}^{+} $,
a positive integer $k$, and disjoint terminal sets $T_{1}, \dots, T_{t} \subseteq V$ with $|T_{i}| \ge 2$ for all $i$,
decide whether there is a subgraph $F$ of $G$ with $\sum_{e \in F} w(e) \le k$
such that each $T_{i}$ is contained in some connected component of $F$.
Note that we can assume that $F$ is a forest.

It is known that \textsc{Steiner Forest} is strongly NP-complete 
(that is, NP-complete even if the weights are given in unary)
on graphs of vertex integrity $5$~\cite{Gassner10}. 
We show that for graphs of small vertex cover number, the problem becomes easier.

Let $G = (V,E)$ be a graph, $F$ a subgraph of $G$, and $S$ a vertex cover of $G$.
We assume without loss of generality that $F$ is a forest.
The following observations follow from this assumption and the fact that $V-S$ is an independent set.
\begin{observation}
  \label{obs:sf-vc-deg2}
  At most $|S| -1$ vertices in $V - S$ have degree 2 or more in $F$.
\end{observation}
\begin{observation}
  \label{obs:sf-vc-cc}
  $F$ has at most $|S|$ connected components.
\end{observation}

\begin{theorem}
\label{thm:sf_vc1}
\textsc{Steiner Forest} can be solved in time $n^{O(\vc)}$,
on $n$-vertex graphs of vertex cover number at most $\vc$.
\end{theorem}
\begin{proof}
Let $G = (V,E)$ be a graph and $S \subseteq V$ be a vertex cover of $G$.
We first guess the set $D \subseteq V-S$ of vertices that have degree at least $2$ in $F$.
By Observation~\ref{obs:sf-vc-deg2}, we know that $|D| \le |S|-1$.
The vertices in $V-(S \cup D)$ have degree at most 1 in $F$:
if such a vertex appears in some $T_{i}$, then it has degree 1 in $F$;
otherwise, it has degree 0 in $F$, and thus can be safely removed from the graph.
We then guess the edges in $F[S \cup D]$.
The number of candidates for $D$ is at most $n^{|S|}$, and 
the number of candidates for the edge set is at most $|S|^{2|S|}$.

We reject the guess $F[S \cup D]$ 
if there are two components of $F[S \cup D]$ that contain elements of $T_{i}$ for some $i$.
Now for each $i$, there is at most one component $C_{i}$ of $F[S \cup D]$
such that $C \cap T_{i} \ne \emptyset$.
If there is no such component, we guess one component of $F[S \cup D]$ from $O(|S|)$ candidates
and call it $C_{i}$. Note that $C_{i}$ and $C_{j}$ may be the same for $i \ne j$.
Now for each $i$ and for each vertex $u \in T_{i} \setminus (S \cup D)$,
we find an edge $\{u,v\}$ of the minimum weight such that $v \in V(C_{i})$
and add $\{u,v\}$ into $F$.
We output a minimum weight forest obtained in this way.
\end{proof}

We denote by \textsc{Unweighted Steiner Forest}
the special case of \textsc{Steiner Forest} such that each edge has weight 1.
By subdividing the edges in the proof in~\cite{Gassner10},
we can show that \textsc{Unweighted Steiner Forest} is NP-complete for graphs of $\tw=3$.

\begin{theorem}
\label{thm:usf-vc}
\textsc{Unweighted Steiner Forest} is fixed-parameter tractable parameterized by $\vc$.
\end{theorem}
\begin{proof}
Let $G = (V,E)$ be the input graph and $S$ be a vertex cover of $G$. Let $s = |S|$.
We reduce the instance by applying the following reduction rules exhaustively.
\begin{enumerate}
  \item If $T_{i}$ contains $s$ or more vertices $v$ in $V-S$ that have the same neighborhood $N(v)$,
  then remove one of them from $T_{i}$ and decrease $k$ by $1$.
  
  \item Let $T_{i(1)}, \dots T_{i(s+1)}$ be $s+1$ distinct terminal sets such that
  $T_{i(j)} \cap S = \emptyset$ for all $j$
  and for each $X \subseteq S$ and $j \ne j'$,
  it holds that $|\{v \in T_{i(j)} \mid N(v) = X\}| = |\{v \in T_{i(j')} \mid N(v) = X\}|$.
  Then replace $T_{i(1)}$ and $T_{i(2)}$ with their union $T_{i(1)} \cup T_{i(2)}$.

  \item If there are two non-terminal vertices of the same neighborhood,
  then remove one of them from the graph.
\end{enumerate}

The first rule is safe by Observation~\ref{obs:sf-vc-deg2}.
At least one of the $s$ vertices of the same neighborhood is a leaf.
Since there are other vertices of the same neighborhood in $T_{i}$,
this leaf can be joined to the component containing the other vertices of $T_{i}$ with an edge.
The safeness of the second rule follows by Observation~\ref{obs:sf-vc-cc}, 
as at least two of them belong to the same connected component.
The third rule is safe because at most one of them is used in an optimal solution.

We can see that if no reduction rule above applies, then the numbers of vertices and of terminal sets 
depend only on $s$.
By the first rule, each terminal set $T_{i}$ contains at most $s \cdot 2^{s}$ vertices.
By the first and second rules, there are at most $s \cdot s^{2^{s}} + s$ terminal sets.
By Reduction rule 3, there are at most $2^{s} + s$ non-terminal vertices.
\end{proof}


\section{Easy problems parameterized by $\td$}
\label{sec:easy-td}

Here we show the following.
\begin{observation}
\textsc{List Hamiltonian Path}, 
\textsc{Directed} $(p,q)$-\textsc{Edge Dominating Set}, and
\textsc{Metric Dimension}
are fixed-parameter tractable parameterized by $\td$.
\end{observation}

\textsc{List Hamiltonian Path} is a generalization of \textsc{Hamiltonian Path}
such that each vertex has a set of permitted positions where it can be put in a Hamiltonian path.
This problem is W[1]-hard parameterized by $\pw$~\cite{MeeksS16}.
By Proposition~\ref{prop:td-path-length},
this problem parameterized by $\td$ admits a trivial FPT algorithm:
if a graph $G$ has at least $2^{\td(G)}$ vertices, then $G$ does not have a Hamiltonian path;
otherwise, we can try all $n! \le (2^{\td(G)})!$ permutations of vertices.

For integers $p,q \ge 0$, the edges \emph{$(p,q)$-dominated} by
an arc $e = (u,v)$ are $e$ itself and all arcs that are 
on a directed path of length at most $p$ to $u$ or
on a directed path of length at most $q$ from $v$.
Then, \textsc{Directed} $(p,q)$-\textsc{Edge Dominating Set} asks 
whether there exists a set $K$ of arcs with $|K| \le k$
such that every arc is $(p,q)$-dominated by $K$.
This problem is W[1]-hard parameterized by $\pw$
but fixed-parameter tractable parameterized by $\tw + p + q$~\cite{BelmonteHK0L18}.
Since we can assume that $p$ and $q$ are smaller than the longest path length,
we can also assume that they are bounded by a function of $\td$.
Thus the fixed-parameter tractability with $\tw + p + q$
implies the fixed-parameter tractability solely with $\td$.
(Here the parameters are defined on the undirected graph obtained by ignoring the directions of arcs.)

The argument for \textsc{Metric Dimension} is slightly more involved.
In \textsc{Metric Dimension},
we are given a graph $G = (V,E)$ and an integer $k$
and asked whether there exists $S \subseteq V$ such that $|S| \le k$
and for each pair $u, v \in V$ there exists $w \in S$ with $\dist(u,w) \ne \dist(v,w)$,
where $\dist(\cdot,\cdot)$ is the distance between two vertices in $G$.
We call such a set $S$ a \emph{resolving set}. 
Recently, this problem is shown to be W[1]-hard parameterized by $\pw$~\cite{BonnetP19}.
Observe that there is an MSO$_{1}$ formula $\varphi(S)$ 
such that its length depends only on the diameter of the underlying graph
and it is evaluated to be \texttt{true} if and only if 
the vertex set $S$ is a resolving set of the underlying graph.
It is known that finding a minimum vertex set $S$ satisfying $\varphi(S)$
is fixed-parameter tractable parameterized by clique-width${}+ |\varphi|$~\cite{CourcelleMR00,Oum08}.
Therefore, \textsc{Metric Dimension} is fixed-parameter tractable parameterized by clique-width${} + {}$the diameter.
Since the clique-width and the diameter are bounded from above by functions of its treedepth,
\textsc{Metric Dimension} is fixed-parameter tractable parameterized by $\td$.
(See \cite{CorneilR05} for an upper bound of clique-width.)


\end{document}